\documentclass[final]{siamltex}
\usepackage{geometry}
\usepackage{graphicx}
\usepackage{amssymb}
\usepackage{amsmath}
\usepackage{epstopdf}
\usepackage{algorithm}
\usepackage{algorithmic}
\usepackage{setspace}
\usepackage{url}
\usepackage{color}
\usepackage{cancel}
\usepackage{enumerate}

\usepackage{caption} 
\usepackage{lscape}

\usepackage{bbm}

\newcommand{\indicator}[1]{\mathbbm{1}_{\left\{ {#1} \right\} }}

\newcommand{\Call}{\mbox{C}}
\newcommand{\Put}{\mbox{P}}
\newcommand{\VIX}{\mbox{VIX}}
\newcommand{\smallVIX}{\mbox{\footnotesize VIX}}

\newcommand{\IFI}{{}_1{F}_1}

\newcommand{\E}{\mathbb E}

\newtheorem{remark}[theorem]{Remark}

\newtheorem{condition}[theorem]{Condition}

\title{Extreme-Strike Comparisons and Structural Bounds for SPX and VIX Options}
\author{A. Papanicolaou\thanks{Department of Finance and Risk Engineering, NYU Tandon School of Engineering, 6 MetroTech Center, Brooklyn, NY 11201 {\em ap1345@nyu.edu}. Part of this research was performed while the author was visiting the Institute for Pure and Applied Mathematics (IPAM), which is supported by the National Science Foundation. A special `thank you' to the associate editor for your consideration and input in the preparation of this article.}}
\begin{document}
\bibliographystyle{alpha}
\maketitle
\begin{abstract}
This article explores the relationship between the SPX and VIX options markets. High-strike VIX call options are used to hedge tail risk in the SPX, which means that SPX options are a reflection of the extreme-strike asymptotics of VIX options, and vice versa. This relationship can be quantified using moment formulas in a model-free way. Comparisons are made between VIX and SPX implied volatilities along with various examples of stochastic volatility models.
\end{abstract}
%%%% for article
%\tableofcontents
%%%%%%%%%%%%%%%%%%%
\section{Introduction}
The S\&P500 (SPX) index and its volatility have been shown to have strong negative correlation. For this reason there is a great deal of interest in the Chicago Board Options Exchange (CBOE) volatility index (VIX), the options-based volatility index and its derivatives for hedging tail risk. In particular, the CBOE has designed the VIX Tail Hedge (VXTH) index based on an SPX and VIX trading strategy, which has been backtested on data and shown to have performed better than the SPX over time periods when there has been a crisis event. Part of the VXTH strategy is to buy high-strike European call options on VIX to insure against losses in SPX, as a large rise in the VIX usually coincides with a drop in the SPX. Prior to VIX derivatives, a similar insurance strategy might have been to buy low-strike European put options on SPX. This similarity means that there is information on the risk-neutral distribution for VIX that is implied by low-strike SPX put options. The markets for SPX and VIX options are very liquid, and so it is useful to have structural bounds that quantify the relationship between the two. In particular, this paper's Lemma \ref{lemma:MGFbound} will show that for discounted SPX price $S_te^{-rt}$ being a local martingale (where $r\geq0$ is the risk-free rate), the moment generating function (MGF) of $\VIX_T^2$ with $\xi\in\mathbb R$ satisfies
\begin{equation}
\label{eq:inequality}
\E_te^{\xi\smallVIX_T^2}\leq \frac1q\E_tS_{T+\tau}^{-\frac{2\xi q}{\tau}}+\frac1p\E_t\left(e^{r\tau}S_T\right)^{\frac{2\xi p}{\tau}}\qquad\hbox{for times $0\leq t\leq T$,}
\end{equation}
where $\E_t$ denotes risk-neutral expectation conditional on the market at time $t$, $\tau=$ 30 days, and both $p\geq1$ and $q\geq 1$ are H\"older-conjugate exponents with $\frac1p+\frac1q=1$. If there exists $\xi>0$ for which the right-hand side of \eqref{eq:inequality} is finite, then the SPX market indicates the VIX's distribution is not \textit{heavy tailed} (i.e. the MGF of $\VIX_T^2$ exists from some $\xi>0$). Conversely, if the MGF $\VIX_T^2$ is infinite for all $\xi>0$, then the right-hand side of \eqref{eq:inequality} is always infinite and there exists no negative moments for $S_{T+\tau}$.

Stochastic volatility and L\'evy jump processes (or a combination of these two) have been used in option pricing since the 1990's, and at that time it may have seemed as if volatility derivatives could be priced and hedged from a well-calibrated model. Indeed, the availability of VIX options data has been an innovation in the study of volatility because it gives new information in addition to the data from SPX options, which means more data to use when calibrating a model. However, data from various days throughout the 2000's exhibit VIX option implied volatilities that do not have a good fit to standards such as the square root (Heston) model. From the perspective of someone searching for the ``right model", model specification remains an important issue because it is a nontrivial task to fit a single model to both the SPX and VIX implied-volatility surfaces. Hence, it would be quite useful if there were a general theory to explain the relationship between markets for VIX and markets for SPX options, and vice versa. This paper presents the beginnings of such a theory in a model-free context.

Of further interest is the understanding of the implied volatility from VIX options. It is certainly true that every asset class requires tailored expert analysis of implied volatility, but VIX option implied volatility is special because it is really the implied \textit{volatility-of-volatility} for the SPX, and hence it is saying something about SPX options. In particular, implied volatilities from VIX options are sometimes very high (i.e. in the range of 80\% for high-strike VIX call options), but there is not yet a standard for making comparisons to implied volatilities observed from SPX options. It would be a significant contribution if implied volatiles from VIX call options could be used to make definitive statements about the no-arbitrage range for SPX implied volatilities. Using the moment formula of \cite{lee2004}, this paper identifies a relationship using extreme-strike asymptotics.

%%%%%%%%%%%%%%%%%%%%%%%%%%%%
\subsection{Literature Review}
The VIX formula (as it has been calculated since 2003) is described in \cite{demeterfi}. A general description of how volatility derivatives are designed and traded is provided in \cite{carrLee2009} with particular attention paid to volatility markets in the era post-2008 crisis. Stochastic and local volatility models are described in \cite{gatheralBook}, and pricing of volatility derivatives based on these models is a standard application of partial differential equation methods. Pricing of VIX options using square-root volatility and jumps is covered in \cite{sepp2008}. An alternative to model-based pricing/hedging are the model-free results found in \cite{carrLee2008,carrLee2010,frizGatheral}; VIX options are priced using non-parametric approximations of the pricing kernel in \cite{dacheng}.

The issues in fitting the Heston model to VIX options are explained in \cite{gatheralSlides}. There has been some success in fitting VIX options to market models of the variance swap term structure (see \cite{carrSun, contKokholm2013}), and the added explanatory power from the inclusion of jumps has been demonstrated in \cite{madanYor2011}. Some studies have shown improved fits to large-strike VIX option implied-volatilities using heavy-tailed process (see \cite{badranBaldeaux2013,gatheral2013,drimus2012}), but heavy tails are not necessarily required as shown in \cite{papanicolaouSircar2014} using a Markov-chain modulation of the Heston model. In  \cite{labordere2014} the links between SPX options and VIX options are studied in a constrained hedging problem, and a link between and SPX and VIX markets is established in \cite{papanicolaou2016} using a time-spread portfolio.

%%%%%%%%%%%%%%%%%%%%%%%%%%%%
\subsection{Main Results Of This Paper}
The main result in this paper is the identification of a link between VIX options and negative moments in the SPX price. In particular, the existence of negative moments in the SPX's risk-neutral distribution is an indication that the VIX's risk-neutral distribution is not heavy tailed. Conversely, if the risk-neutral distribution of $\VIX_T^2$ is heavy tailed, then the price for SPX has no negative moments. The results can be considered model free, as the main assumptions are (i) absence of arbitrage and (ii) $e^{-rt}S_t$ is a continuous local (super) martingale. The majority of the calculations are made under the assumption that prices are given by risk-neutral expectations, with the exception of Section \ref{sec:modelMisprice} where mispriced options are shown to be arbitrageable if they can be used to construct replicating portfolios that violate the inequality of \eqref{eq:inequality}. The paper provides a detailed application of the theory to specific models that are frequently used in SPX and VIX options pricing, with an assessment of their relative usefulness based on historically observed market behavior of these options.

The rest of the paper is organized as follows: Section \ref{sec:framework} introduces the probabilistic framework and describes how to price options on SPX and VIX; Section \ref{sec:extremeStrikes} presents the main results (including Lemma \ref{lemma:MGFbound}) and other ideas relevant to the problem; Section \ref{sec:modelMisprice} gives static arbitrage portfolios that can be implemented if the market data is mispriced in such a way so that the inequality of Lemma \ref{lemma:MGFbound} is reversed; Section \ref{sec:examples} presents various stochastic volatility models and discusses how each relates to the paper's results. Section \ref{sec:conclusions} concludes.

%%%%%%%%%%%%%%%%%%%%%%%%%%%%%%%%%%%%%
\section{Probabilistic Framework for Pricing}
\label{sec:framework}
Let $S_t$ denote the price of the SPX at some time $t\geq 0$. The model considered throughout this paper has an asset whose log returns are given by a stochastic volatility model,
\begin{equation}
\label{eq:returns}
d\log(S_t)=\left(r-\frac12\sigma_t^2\right)dt+ \sigma_tdW_t
\end{equation}
where $r\geq 0$ is the risk-free rate of interest, $W$ is a risk-neutral Brownian motion, and $\sigma_t$ is a volatility process that is right continuous with left-hand limits and non-anticipative of $W$. 
\begin{condition}[Finite Second Moment of Stochastic Integral]
\label{cond:int2}
For $T<\infty$, the second moment of stochastic integral $\int_0^T\sigma_udW_u$ is finite,

\[\E \left(\int_0^T\sigma_udW_u\right)^2= \E \int_0^T\sigma_u^2du<\infty\ .\]
\end{condition}
Condition \ref{cond:int2} implies that $\int_0^t\sigma_udW_u$ is a true, square integrable martingale on finite time-interval $[0,T]$, but $S_te^{-rt}$ may still be a strict local martingale. The \textit{Novikov Condition} ensures for $T<\infty$ that the process $S_te^{-rt}$ is a true martingale on $[0,T]$ if $\E e^{\frac12\int_0^T\sigma_u^2du}<\infty$. The Novikov condition is very strong and often doesn't hold for stochastic volatility models. Other conditions for exponential martingales are discussed in \cite{klebaner2012}. This paper will rely on Condition \ref{cond:int2}, will not assume Novikov, and will show the martingale property on a case-by-case basis. 

Another important condition is the existence of $S_T$'s negative moments: 
\begin{condition}[Negative Moments]
\label{cond:moments}
For $T<\infty$, there exists constant $q>0$ such that
\[\E S_T^{-q}<\infty\ .\]
\end{condition}
For $S_t$ a supermartingale, Condition \ref{cond:moments} implies existence of the MGF of $\log(S_T)$ over a finite interval containing zero,
\[\E e^{\xi\log(S_T)}<\infty\qquad\forall \xi\in[-q,1]\ .\]
Condition \ref{cond:moments} is used in \cite{lee2004} to obtain small-strike bounds on implied volatility. In particular, the supremum over all $\{q>0:\E S_T^{-q}<\infty\}$ is identified with the asymptotic rate at which implied volatility grows as strike-price goes to zero for options on $S_T$; this is part 2 of the moment formula \cite{lee2004} that will be reviewed in Section \ref{sec:momentFormulas}.

%%%%%%%%%%%%%%%%%%%%%%%%%%%%%%%%%%%%%
\subsection{Variance Swaps and the VIX Index}
Consider European call and put options on $S_T$ for some fixed time $T\in(0,\infty)$ and some fixed strike $K\in[0,\infty)$, both of which are processes
\begin{align*}
\Call(t,K,T)&\triangleq B_{t,T}\E _t(S_T-K)^+\\
\Put(t,K,T)&\triangleq B_{t,T}\E _t(K-S_T)^+
\end{align*} 
for some $0\leq t\leq T$, where $B_{t,T}= e^{-r(T-t)}$ is the discount factor and the expectation operator is defined as $\E _t\triangleq \E [~\cdot~|\mathcal F_t] $ with $\mathcal F_t$ denoting a $\sigma$-algebra with respect to which $W$ is Brownian motion, $\sigma_t$ is $\mathcal F_t$ adapted, and $S_0$ is $\mathcal F_0$ adapted. Throughout the paper, an expectation without a subscript is conditional at time $t=0$, that is, $\E  = \E _0$. 

\begin{definition}[SPX Implied Volatility $\hat\sigma$]
\label{def:hatSigma}
Implied volatility for SPX options is denoted with $\hat \sigma(t,K,T)$ and is the unique volatility input to the Black-Scholes prices such that
\begin{align*}
&\Put(t,K,T) =B_{t,T}\left( \Phi\left(-d_-\right)K-\Phi\left(-d_+\right)\E_tS_T\right)
\end{align*}
where $d_\pm = \frac{\log(\E_tS_T/K)}{\hat\sigma(t,K,T)\sqrt{T-t}}\pm\tfrac{\hat\sigma(t,K,T)\sqrt{T-t}}{2}$ and $\Phi$ is the standard normal cumulative distribution function. The quantity $\hat\sigma$ could be equivalently redefined using call options via the put-call parity.
\end{definition}

A variance swap for the time period $[t,T]$ with $t<T$ has a floating leg of $\frac{1}{T-t}\int_t^T\sigma_u^2du$ (equal to the quadratic variation of $\log(S_t)$ divided by time) and a fixed leg that is chosen such that the contract has zero entry cost at time $t$. This fixed leg is the \textbf{variance-swap rate:}
\begin{equation*}
\mbox{variance-swap rate}= \E _t\left[\frac{1}{T-t}\int_t^T\sigma_u^2du\right]\ .
\end{equation*}
When trading in variance swaps, an important instrument is the log contract with time-$T$ payout of $\log(S_T/\E _tS_T)$. As shown in \cite{demeterfi}, the negative log contract is replicated by a portfolio of European call and put options by taking the expectation of the identity 
\[-\log(S_T/s^*) = -\frac{S_T-s^*}{s^*}+\int_0^{s^*}\frac{(K-S_T)^+}{K^2}dK+\int_{s^*}^\infty\frac{(S_T-K)^+}{K^2}dK\ ,\]
which holds for any reference point $s^*>0$. Taking $s^*=\E _tS_{t+\tau}$ yields the VIX formula: 
\begin{equation}
\label{eq:VIX}
\VIX_t= \sqrt{\frac{2}{\tau B_{t,t+\tau}}\left(\int_0^{\E _tS_{t+\tau}}\Put(t,K,t+\tau)\frac{dK}{K^2}+\int_{\E _tS_{t+\tau}}^\infty \Call(t,K,t+\tau)\frac{dK}{K^2} \right)}\ ,
\end{equation}
where $\tau=$ 30 days. By definition, equation \eqref{eq:VIX} is the square root of the log contract's price $\VIX_t=\sqrt{-\frac{2}{\tau}\E _t\log\left(S_{t+\tau}\Big/\E _tS_{t+\tau}\right)}$. By assuming Condition \ref{cond:int2} for the continuous model in \eqref{eq:returns}, the risk-neutral price of the log contract is equal to the variance-swap rate, and hence the VIX index is the square root of the variance-swap rate for the coming 30 days,
\begin{equation}
\label{eq:noMultiplier}
\left(\hbox{Condition \ref{cond:int2}}\right)\Rightarrow \VIX_t=\sqrt{\hbox{variance-swap rate}} \ .
\end{equation}

%%%%%%%%%%%%%%%%%%%%%%%%%%%%%%%%%%%%%
\subsection{VIX Future and VIX Options}
Define the future contract on $\VIX_T$ at time $t\leq T$ as
\begin{equation}
\label{eq:vixFuture}
X_{t,T}=\E _t\VIX_T= \E _t\sqrt{-\frac{2}{\tau}\E _T\log\left(S_{T+\tau}\Big/\E _TS_{T+\tau}\right)}\ .
\end{equation}
The price $X_{t,T}$ is important in the VIX market because (unlike the VIX index) it is a trade-able asset. European call and put options on the VIX are the expectation of functions of $\VIX_T$, but should be thought of as options on $X_{T,T}$,
\begin{align*}
\Call^{vix}(t,K,T)&\triangleq B_{t,T} \E_t(\VIX_T-K)^+=B_{t,T}\E_t(X_{T,T}-K)^+\\
\Put^{vix}(T,K,T)&\triangleq B_{t,T}\E_t(K-\VIX_T)^+=B_{t,T}\E_t(K-X_{T,T})^+\ .
\end{align*}
Considering these options as payoffs on $X_{T,T}$ makes more clear the convention for $\Delta$-hedging VIX options with the future $X_{t,T}$. It is also the convention for VIX options to quote implied volatility by inverting the Black-Scholes formula on the VIX future, as is also done in \cite{papanicolaouSircar2014}. 

\begin{definition}[VIX Implied Volatility $\hat\nu$]
\label{def:hatNu}
Implied volatility for VIX options is denoted with $\hat \nu(t,K,T)$ and is the unique volatility input to the Black-Scholes prices such that
\begin{align*}
&\Call^{vix}(t,K,T) = B_{t,T}\left(\Phi\left(d_+\right)X_{t,T}-\Phi\left(d_-\right)K\right)\ ,
\end{align*}
where $d_\pm = \frac{\log(X_{t,T}/K)}{\hat\nu(t,K,T)\sqrt{T-t}}\pm\tfrac{\hat\nu(t,K,T)\sqrt{T-t}}{2}$ and $\Phi$ is the standard normal cumulative distribution function. The quantity $\hat \nu$ could be equivalently redefined using VIX put options via put-call parity. 
\end{definition}

Figures \ref{fig:spxImpliedVol} and \ref{fig:vixImpliedVol} show implied volatility for SPX and VIX options for September 9th of 2010, a day during the European debt crisis when options were trading with high implied volatility. Notice the right-hand skew of the $\hat\nu$ in Figure \ref{fig:vixImpliedVol}, which corresponds to volatility tail risk and is a stylistic feature of VIX options that should be captured by a stochastic volatility model that aims to price VIX options in periods of higher volatility (see \cite{drimus2012, gatheral2013,papanicolaouSircar2014}).

\begin{figure}[htbp] %  figure placement: here, top, bottom, or page
   \centering
   \includegraphics[width=4.3in]{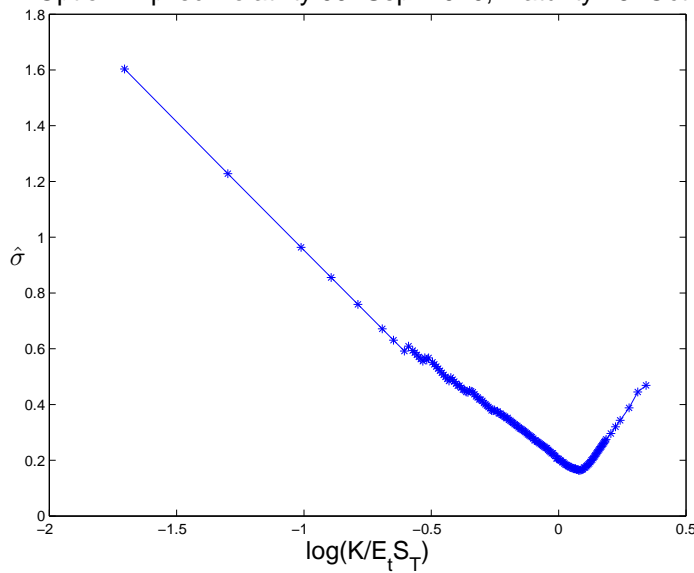} 
   \caption{Implied volatility for SPX put options on September 9th 2010 and maturity on October 16th 2010. The SPX future price is $\E _tS_T = 1101.97$ and the risk-free rate is approximately $r=.28\%$. The left-hand skew observed in this figure has been a common sight since the late 1980's. Stochastic volatility models (such as the Heston model) and various L\'evy models have successfully fit this skew. However, VIX options introduce another skew that is a derivative of the SPX skew (see Figure \ref{fig:vixImpliedVol}), and this new skew has forced the re-evaluation of standards in stochastic volatility models.}
   \label{fig:spxImpliedVol}
\end{figure}

\begin{figure}[htbp] %  figure placement: here, top, bottom, or page
   \centering
   \includegraphics[width=4.3in]{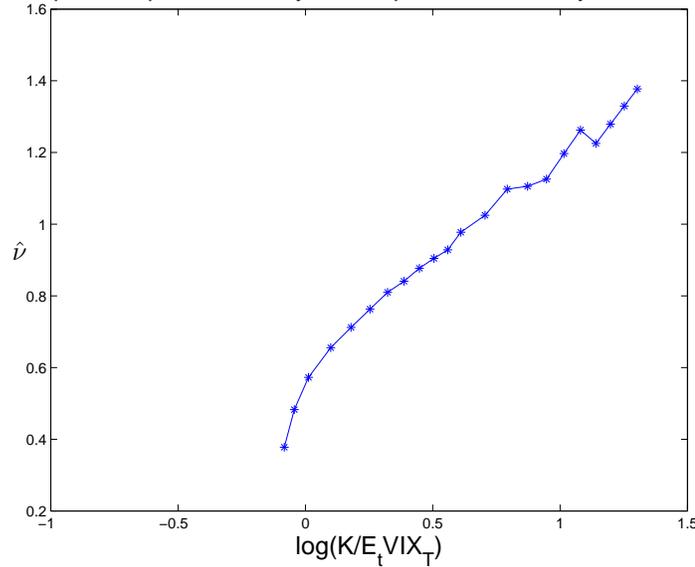} 
   \caption{Implied volatility for VIX call options on September 10th 2010, maturity on October 20th 2010. The VIX future price is $X_{t,T}=\E _t\VIX_T = 27.75\%$. Data from this day highlights the commonly-observed right skew from VIX options. This is the stylistic feature that is identified in the literature (see \cite{drimus2012, gatheralSlides,papanicolaouSircar2014}). Right skew will affect how selection of a stochastic volatility model is made, because it suggests that heavy-tailed volatility models (such as the 3/2 model or double log-normal model in \cite{gatheral2013}) are better-suited than the light-tailed Heston model.}
   \label{fig:vixImpliedVol}
\end{figure}

%%%%%%%%%%%%%%%%%%%%%%%%%%%%%%%%%%%%%%%%%%%%%%%%%%%%%%%%%%%%%%%%%%%%%%%%%%%%%%%%%%%%%%%%%%%%
%%%%%%%%%%%%%%%%%%%%%%%%%%%%%%%%%%%%%%%%%%%%%%%%%%%%%%%%%%%%%%%%%%%%%%%%%%%%%%%%%%%%%%%%%%%%
%%%%%%%%%%%%%%%%%%%%%%%%%%%%%%%%%%%%%%%%%%%%%%%%%%%%%%%%%%%%%%%%%%%%%%%%%%%%%%%%%%%%%%%%%%%%
%%%%%%%%%%%%%%%%%%%%%%%%%%%%%%%%%%%%%%%%%%%%%%%%%%%%%%%%%%%%%%%%%%%%%%%%%%%%%%%%%%%%%%%%%%%%
\section{Extreme-Strike Asymptotics}
\label{sec:extremeStrikes}
Results that are considered model free will usually require some assumptions, such as the VIX being finite almost surely, a condition that is ensured by Condition \ref{cond:int2}. Indeed, Condition \ref{cond:int2} is the key assumption for the square of the VIX formula in equation \eqref{eq:VIX} to be equal to the variance-swap rate, but it turns out that the negative moments described in Condition \ref{cond:moments} are sufficient to imply Condition \ref{cond:int2}. The following proposition proves this statement and is the first instance in this paper where a connection is identified between the VIX and negative moments in SPX:

\begin{proposition} 
\label{prop:negMoments_squareIntegrability}
Let $S_tB_{0,t}$ be a supermartingale on $[0,T]$ with $\int_0^T\sigma_t^2dt<\infty$ a.s. and satisfying Condition \ref{cond:moments}. Then Condition \ref{cond:int2} holds. 

\end{proposition}

\begin{proof}
The stochastic integral $\int_0^t\sigma_udW_u$ is a local martingale, and so there exists a family of increasing stopping times $(\mathcal T_n^1)_{n=1,2,3,\dots}$ such that $\int_0^{t\wedge\mathcal T_n^1}\sigma_udW_u$ is a martingale for all $n<\infty$ and $\mathcal T_n^1\wedge t\nearrow t$ a.s. as $n\rightarrow \infty$. There also exists an increasing family $(\mathcal T_n^2)_{n=1,2,3,\dots}$ such that $\int_0^{t\wedge\mathcal T_n^2}\sigma_uS_udW_u$ is a martingale for all $n<\infty$, and $\mathcal T_n^2\wedge t\nearrow t$ a.s. as $n\rightarrow \infty$, and hence $\E S_{t\wedge\mathcal T_n^2}=S_0\E1/ B_{0,t\wedge\mathcal T_n^2}$. Then defining $\mathcal T_n=\mathcal T_n^1\wedge\mathcal T_n^2$, the expected variance of the stopped process can be bounded from above,
\begin{align*}
\frac12\E\int_0^{T\wedge\mathcal T_n}\sigma_t^2dt 
&=\E\int_0^{T\wedge\mathcal T_n}\frac{dS_t}{S_t}-\log(\E S_{T\wedge \mathcal T_n}/ S_0)  -\E\log(S_{T\wedge\mathcal T_n}/\E S_{T\wedge \mathcal T_n})\\
&=r\E\int_0^{T\wedge\mathcal T_n}dt -\log(\E e^{r(T\wedge \mathcal T_n)})  \\
&~~~~~+\int_0^{\E S_{T\wedge\mathcal T_n}}\frac{\E(K-S_{T\wedge\mathcal T_n})^+}{K^2}dK+\int_{\E S_{T\wedge\mathcal T_n}}^\infty\frac{\E(S_{T\wedge\mathcal T_n}-K)^+}{K^2}dK\\
&\leq \underbrace{r\E\int_0^{T\wedge\mathcal T_n}dt -r\E (T\wedge \mathcal T_n)}_{=0}  \\
&~~~~~+\int_0^{\E S_{T\wedge\mathcal T_n}}\frac{\E(K-S_{T\wedge\mathcal T_n})^+}{K^2}dK+\int_{\E S_{T\wedge\mathcal T_n}}^\infty\frac{\E(S_{T\wedge\mathcal T_n}-K)^+}{K^2}dK\ ,
\end{align*}
where the last line comes from Jensen's inequality. Using the bound $\E(S_{T\wedge\mathcal T_n}-K)^+\leq \E S_{T\wedge\mathcal T_n}$, the above quantity can be further estimated,
\begin{align*}
\frac12\E\int_0^{T\wedge\mathcal T_n}\sigma_t^2dt &\leq \int_0^{\E S_{T\wedge\mathcal T_n}}\frac{\E(K-S_{T\wedge\mathcal T_n})^+}{K^2}dK+\E S_{T\wedge\mathcal T_n}\int_{\E S_{T\wedge\mathcal T_n}}^\infty\frac{1}{K^2}dK\\
&=\int_0^{\E S_{T\wedge\mathcal T_n}}\frac{\E(K-S_{T\wedge\mathcal T_n})^+}{K^2}dK+1\ .
\end{align*}
On the left-hand side $\lim_{n\rightarrow\infty}\E\int_0^{T\wedge\mathcal T_n}\sigma_t^2dt=\E\int_0^T\sigma_t^2dt$ by monotone convergence theorem. On the right-hand side a bound comes from the estimate $\E_t(K-S_{T\wedge\mathcal T_n})^+ \leq \frac{\E _tS_{T\wedge\mathcal T_n}^{-q}}{q+1}\left(\frac{q}{q+1}\right)^{q}K^{1+q}$
for all $K\in[0,\infty)$ and $q>0$ (see \cite{lee2004}), along with the convexity inequality 
\[\frac{1}{(B_{T\wedge\mathcal T_n,T}S_T)^q} \geq \frac{1}{S_{T\wedge\mathcal T_n}^q} - q \frac{B_{T\wedge\mathcal T_n,T}S_T-S_{T\wedge\mathcal T_n}}{S_{T\wedge\mathcal T_n}^{q+1}}\ ,\]
with the expectation being bounded,
\begin{align*}
\E S_{T\wedge\mathcal T_n}^{-q} &\leq \E\left[(B_{T\wedge\mathcal T_n,T}S_T)^{-q} + q (B_{T\wedge\mathcal T_n,T}S_T-S_{T\wedge\mathcal T_n})S_{T\wedge\mathcal T_n}^{-(q+1)}\right]\\
&=\E\left[(B_{T\wedge\mathcal T_n,T}S_T)^{-q} + q \underbrace{\E\left[(B_{T\wedge\mathcal T_n,T}S_T-S_{T\wedge\mathcal T_n})\Big|\mathcal F_{T\wedge\mathcal T_n}\right]}_{\leq 0~~\hbox{supermartingale}}S_{T\wedge\mathcal T_n}^{-(q+1)}\right]\\
&\leq B_{0,T}^{-q}\E S_T^{-q}\ ,
\end{align*}
and hence there is the bound,
\begin{align*}
\int_0^{\E S_{T\wedge\mathcal T_n}}\frac{\E_t(K-S_{T\wedge\mathcal T_n})^+}{K^2}&\leq\frac{B_{0,T}^{-q}\E _tS_T^{-q}}{q+1}\left(\frac{q}{q+1}\right)^{q}\int_0^{\E S_{T\wedge\mathcal T_n}}K^{-1+q}dK\\
&\leq  \frac{B_{0,T}^{-q}\E _tS_T^{-q}}{q(q+1)}\left(\frac{q}{q+1}\right)^{q}(S_0/B_{0,T})^q\ ,
\end{align*}
which is finite for some $q>0$ by Condition \ref{cond:moments}. Hence there is the inequality
\begin{align*}
&\frac12\E\int_0^T\sigma_t^2dt
\leq \frac{\E S_T^{-q}}{q(q+1)}\left(\frac{q}{q+1}\right)^q(S_0/B_{0,T}^2)^q+1 \ ,
\end{align*}
and if $\E S_T^{-q}<\infty$ then Condition \ref{cond:int2} holds. 
\end{proof}

Another instance where finite SPX moments are important is in determining the existence of the MGF of $\VIX_T^2$. The following lemma will be used throughout the rest of the paper:

\begin{lemma} 
\label{lemma:MGFbound}
Let $S_tB_{0,t}$ be a supermartingale on $[0,T+\tau]$ satisfying Condition \ref{cond:int2}. For any $\xi\in\mathbb R$, the MGF $\E _te^{\xi \smallVIX_T^2} $ satisfies the inequality
\begin{equation}
\label{eq:MGFbound}
\E _te^{\xi \smallVIX_T^2} \leq\frac1q\E _t\left(S_{T+\tau}\right)^{-\frac{2\xi q}{\tau}}+\frac1p\E _t\left(S_T/B_{T,T+\tau}\right)^{\frac{2\xi p}{\tau}}\qquad\forall t\leq T\ ,
\end{equation}
where $p>1$ and $q>1$ with $\frac1p+\frac1q=1$ (i.e. $p$ and $q$ are conjugate exponents). This is a strict inequality if $\E_t\left(\frac{S_{T+\tau}}{\E_TS_{T+\tau}}-1\right)^2>0$.
\end{lemma}
\begin{proof}
From Jensen's and Young's inequality, 
\begin{align*}
\E _te^{\xi \smallVIX_T^2} &=\E _te^{-\frac{2\xi}{\tau}\E _T\log(S_{T+\tau}/\E_TS_{T+\tau})}\\
&\leq\E _te^{\log\left(\left(\frac{S_{T+\tau}}{\E_TS_{T+\tau}}\right)^{-\frac{2\xi}{\tau}}\right)}\hspace{2cm}\hbox{(Jensen's inequality)}\\
&=\E _t\left(\frac{S_{T+\tau}}{\E_TS_{T+\tau}}\right)^{-\frac{2\xi}{\tau}}\\
&\leq\frac1q\E _t\left(S_{T+\tau}\right)^{-\frac{2\xi q}{\tau}}+\frac1p\E _t\left(\frac{1}{\E_TS_{T+\tau}}\right)^{-\frac{2\xi p}{\tau}}\hspace{.1cm}\hbox{(Young's inequality)}\\
&\leq\frac1q\E _t\left(S_{T+\tau}\right)^{-\frac{2\xi q}{\tau}}+\frac1p\E _t\left(S_T/B_{T,T+\tau}\right)^{\frac{2\xi p}{\tau}}\ .
\end{align*}
In this case, Jensen's inequality is an equality iff the random variable has zero variance (i.e. if $\E_t\left(\frac{S_{T+\tau}}{\E_TS_{T+\tau}}-1\right)^2=0$), and hence the inequality is strict in non-degenerate cases.
\end{proof}

Lemma \ref{lemma:MGFbound} is a useful tool when evaluating the market for VIX options, primarily because it shows how existence of a negative moment $\E _tS_{T+\tau}^{-q}$ for some $q>0$ implies that the VIX-squared process is not heavy tailed. Conversely, if the MGF $\E _te^{\xi \smallVIX_T^2} =\infty$ for all $\xi>0$, then $\E_tS_{T+\tau}^{-q}=\infty $ for all $q>0$. In both cases, for $\frac{2\xi p}{\tau}\leq 1$ the (super) martingale property ensures $\E _t\left(S_T/B_{T,T+\tau}\right)^{\frac{2\xi p}{\tau}}<\infty$, so that $\E _t\left(S_{T+\tau}\right)^{-\frac{2\xi q}{\tau}}<\infty$ implies $\E _te^{\xi \smallVIX_T^2} <\infty$, and $\E _te^{\xi \smallVIX_T^2} =\infty$ implies $\E _t\left(S_{T+\tau}\right)^{-\frac{2\xi q}{\tau}}=\infty$.

%%%%%%%%%%%%%%%%%%%%%%%%%%%%%%%%%%%%%
\subsection{Moment Formulas}
\label{sec:momentFormulas}
The Moment Formula from \cite{lee2004} consists of parts 1 and part 2 describing the right and left tail, respectively, of the implied volatility smile. Let the price process $S_t$ be a martingale and define $\tilde p \triangleq \sup\{p\geq0|\E _tS_T^{1+p}<\infty\}$. Part 1 states that

\begin{equation}
\label{eq:moment1}
 \limsup_{K\nearrow\infty}\frac{\hat\sigma^2(t,K,T)}{\log(K/\E_tS_T)/(T-t)}=\beta_R\in[0,2]\ ,
 \end{equation}
 where $\tilde p = \frac{1}{2\beta_R}+\frac{\beta_R}{8}-\frac12$. Equivalently, $\beta_R = 2-4\left(\sqrt{\tilde p^2+\tilde p}-\tilde p\right)$ with $\beta_R=0$ when $\tilde p=\infty$. Next define $\tilde q \triangleq \sup\{q\geq0|\E _tS_T^{-q}<\infty\}$. Part 2 states that
\begin{equation}
\label{eq:moment2}
\limsup_{K\searrow 0}\frac{\hat\sigma^2(t,K,T)}{\log(K/\E_tS_T)/(T-t)}=\beta_L \in[0,2]
\end{equation}
where $\tilde q = \frac{1}{2\beta_L}+\frac{\beta_L}{8}-\frac12$. Equivalently, $\beta_L = 2-4\left(\sqrt{\tilde q^2+\tilde q}-\tilde q\right)$ with $\beta_L=0$ when $\tilde q=\infty$. Parts 1 and 2 both take $1/0\triangleq\infty$. For many models the limit supremum in \eqref{eq:moment1} and \eqref{eq:moment2} can be replaced with a proper limit (see \cite{friz2008}). Figure \ref{fig:prelim_momentFormula} shows how the moment formulas apply to the data with $\beta_L$ and $\beta_R$ estimated from the most extreme strikes in the September 2010 options data seen Figures \ref{fig:spxImpliedVol} and \ref{fig:vixImpliedVol}.
\begin{figure}[htbp] %  figure placement: here, top, bottom, or page
\centering
\begin{tabular}{cc}
   \includegraphics[width=2.9in]{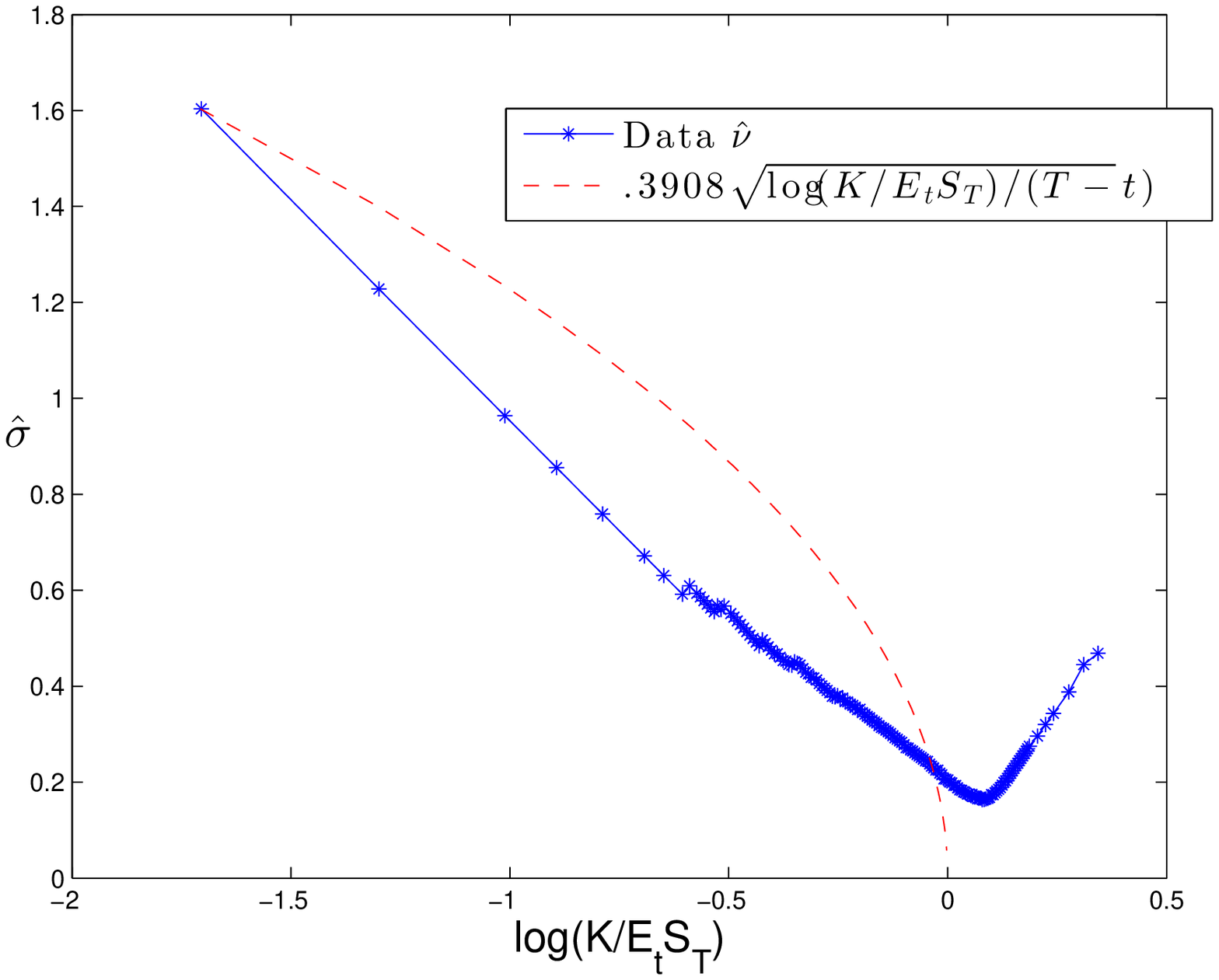} &
   \includegraphics[width=2.9in]{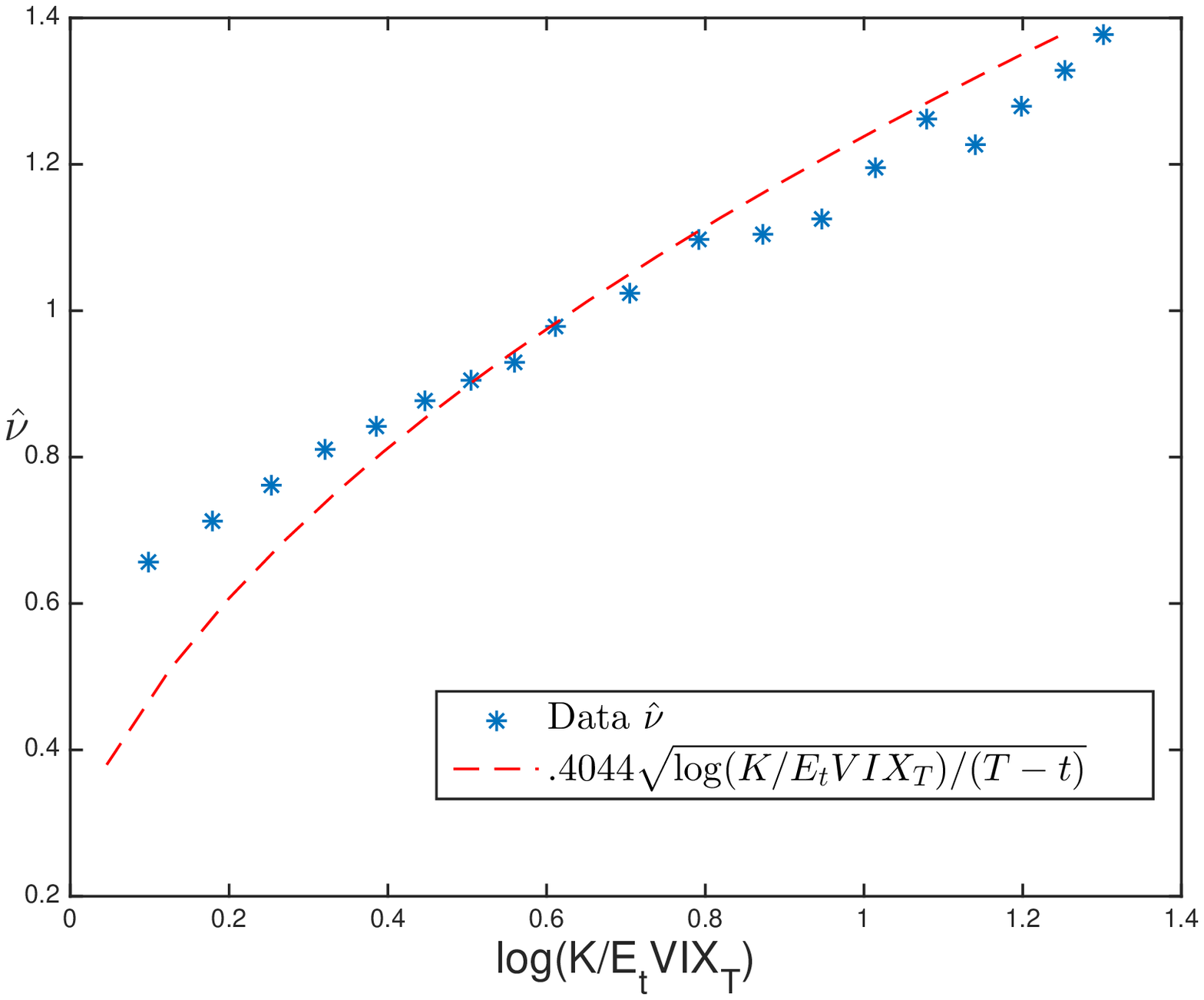} 
   \end{tabular}
   \caption{\textbf{Left:} From the lowest-strike SPX put option there is the estimate $\beta_L\geq 0.3908$, which is the coefficient of the dashed line that is the extreme-strike asymptotic from the moment formula. \textbf{Right:} From the highest-strike VIX option there is the estimate $\sqrt{\beta_R^{vix}}\geq .4044$.}
   \label{fig:prelim_momentFormula}
\end{figure}

Moment formulas can be used to show how moment explosion in the VIX options market affects implied volatility in SPX options. Define $\tilde p^{vix} = \sup\left\{p>0\Big|\E _t\VIX_T^{1+p}<\infty\right\}$.
If $\tilde p^{vix}<\infty$ then the MGF $\E _te^{\xi \smallVIX_T^2} =\infty$ for all $\xi>0$ and by Lemma \ref{lemma:MGFbound} it follows that $\E _tS_{T+\tau}^{-q}=\infty$ for all $q>0$. Hence by part 2 of the moment formula, as stated in equation \eqref{eq:moment2}, it follows that $\beta_L=2$ for SPX options with exercise at time $T+\tau$, and the implied volatility limit is at its maximum, $\limsup_{K\searrow 0}\frac{\hat \sigma^2(t,T+\tau,K)}{\log(K/\E_tS_{T+\tau})/(T+\tau-t)}=2$.

Similarly, moments of SPX's distribution can say something about implied volatility of VIX options. Define $\beta_R^{vix} =  \limsup_{K\nearrow\infty}\frac{\hat\nu^2(t,K,T)}{\log(K/\E_tS_T)/(T-t)}$, and suppose $\tilde q = \sup\left\{q>0\Big|\E _tS_{T+\tau}^{-q}<\infty\right\}>0$. Then from Lemma \ref{lemma:MGFbound} it follows that $\E _te^{\xi\VIX_T^2}<\infty$ for some $\xi>0$. Moreover, if $\VIX_T^2$ has finite MGF for positive $\xi$ then $\E _t\VIX_T^n<\infty$ for all $n>0$, and then by equation \eqref{eq:moment1} it follows that $\beta_R^{vix}=0$, giving the extreme-strike asymptotic $\limsup_{K\nearrow \infty}\frac{\hat \nu^2(t,T,K)}{\log(K/X_{t,T})/(T-t)}=0$. Moreover, finite MGF for $\VIX_T^2$ for positive $\xi$ is equivalent to saying that the VIX-squared does not have a heavy-tailed distribution, and hence $\tilde q>0$ implies $\VIX_T$ does not have heavy tails.

 More generally, the moment formula and Lemma \ref{lemma:MGFbound} are used to show how implied volatility from VIX options gives a lower bound on the implied volatilities of SPX options.  
 \begin{proposition}
 \label{prop:impVolLowerBound}
Assume Condition \ref{cond:int2} and let 
 \[\tilde \xi=\sup\left\{\xi\geq0\Big|\E _te^{\xi \smallVIX_T^2} <\infty\right\}.\]
 \begin{enumerate}
 \item
If $\tilde\xi<\infty$ and $\E _tS_T^{\frac{4\tilde \xi}{\tau}}<\infty$, then 
\[\tilde q = \sup\left\{q\geq0\Big|\E _tS_{T+\tau}^{-q}<\infty\right\}\leq\frac{4\tilde \xi}{\tau} ,\]
and 
\begin{align*}
&\limsup_{K\searrow 0}\frac{\hat\sigma^2(t,K,T+\tau)}{\log(K/\E_tS_{T+\tau})/(T+\tau-t)}=\beta_L\geq 2-4\left(\sqrt{\left(\frac{4\tilde \xi}{\tau}\right)^2+\frac{4\tilde \xi}{\tau}}-\frac{4\tilde \xi}{\tau}\right)\geq0\ .
\end{align*}
\item
If $\tilde\xi<\infty$ and $\E _tS_{T+\tau}^{\frac{-4\tilde \xi}{\tau}}<\infty$, then $\frac{4\tilde \xi}{\tau}\geq1$ and 
\[\tilde p = \sup\left\{p\geq0\Big|\E _tS_T^{1+p}<\infty\right\}\leq\frac{4\tilde \xi}{\tau}-1 ,\]
and 
\begin{align*}
&\limsup_{K\nearrow \infty}\frac{\hat\sigma^2(t,K,T)}{\log(K/\E_tS_T)/(T-t)}=\beta_R\\
&\hspace{3cm}\geq 2-4\left(\sqrt{\left(\frac{4\tilde \xi}{\tau}-1\right)^2+\frac{4\tilde \xi}{\tau}-1}-\left(\frac{4\tilde \xi}{\tau}-1\right)\right)\\
&\hspace{3cm}= -2-4\left(\sqrt{\left(\frac{4\tilde \xi}{\tau}\right)^2-\frac{4\tilde \xi}{\tau}}-\frac{4\tilde \xi}{\tau}\right)\geq0\ .
\end{align*}
\end{enumerate}
 \end{proposition}
 \begin{proof}Part 1: If $\tilde\xi<\infty$ then for any $\epsilon\in(0,1)$ it is the case that $\E _te^{\tilde \xi(1+\epsilon) \smallVIX_T^2} =\infty$. Using Lemma \ref{lemma:MGFbound} with $p=2/(1+\epsilon)$ and $q=2/(1-\epsilon)$, if $\E _tS_T^{\frac{4\tilde \xi}{\tau}}<\infty$ it then follows that $\E _tS_{T+\tau}^{-\frac{4\tilde \xi(1+\epsilon)}{\tau(1-\epsilon)}}=\infty$ and therefore $\tilde q = \sup\{q\geq0|\E _tS_{T+\tau}^{-q}<\infty\}< \frac{4\tilde \xi(1+\epsilon)}{\tau(1-\epsilon)}$ for all $\epsilon\in(0,1)$, and hence $\tilde q\leq \frac{4\tilde \xi}{\tau}$. Now define the function 
 \[f(q) = 2-4\left(\sqrt{q^2+q}-q\right)\]
 and notice that $f'(q)<0$ for all $q>0$, so that in equation \eqref{eq:moment2} the limit supremum is $\beta_L=f(\tilde q)\geq f\left(\frac{4\tilde \xi}{\tau}\right)$.
 
Part 2: If $\tilde\xi<\infty$ then for any $\epsilon\in(0,1)$ it is the case that $\E _te^{\tilde \xi(1+\epsilon) \smallVIX_T^2} =\infty$. Using Lemma \ref{lemma:MGFbound} with $p=2/(1-\epsilon)$ and $q=2/(1+\epsilon)$, if $\E _tS_{T+\tau}^{\frac{-4\tilde \xi}{\tau}}<\infty$ then it follows that $\E _tS_T^{\frac{4\tilde \xi(1+\epsilon)}{\tau(1-\epsilon)}}=\infty$, therefore $1+\tilde p< \frac{4\tilde \xi(1+\epsilon)}{\tau(1-\epsilon)}$ for all $\epsilon\in(0,1)$, and hence $\tilde p\leq \frac{4\tilde \xi}{\tau}-1$. The remainder of the proof is similar to the argument used in Part 1, except with an application of equation \eqref{eq:moment1}.
 \end{proof}

%%%%%%%%%%%%%%%%%%%%%%%%%%%%
\subsection{Replication of VIX MGF with VIX Options}
\label{sec:hedging}

From a market of VIX options with a continuum of strikes comes a tremendous amount of information about the SPX. In particular, via the Breeden and Litzenberger formula \cite{breedenLitzenberger1978} one obtains a risk-neutral distribution on the portfolio of calls and puts given in equation \eqref{eq:VIX}. Via integration-by-parts it is possible to replicate risk-neutral expectations using European call and put options (see \cite{bick1982,carrMadan2001}). This section uses such techniques to show how the MGF of $\VIX_T^2$ can be replicated, and if $\E_te^{\xi\VIX_T^2}<\infty$ for some $\xi>0$ then VIX call options decay quickly as $K$ grows.

Moments on VIX can be replicated with VIX options: for positive $n> 1$
\[B_{t,T}\E _t\VIX_T^n=
n(n-1)\int_0^\infty K^{n-2}\Call^{vix}(t,K,T)dK\ .\]
Furthermore, the MGF of $\VIX_T^2$ can be replicated, and if finite will give a rate at which call options must decay for large strikes.

\begin{proposition}
\label{prop:MGFhedge}
Suppose Condition \ref{cond:int2}. For any $\xi\in\mathbb R$ there is a replication of $e^{\xi \smallVIX_T^2}$ that gives the MGF of $\VIX_T^2$ in terms of VIX options,
\[\E _te^{\xi \smallVIX_T^2} =1+ \frac{1}{B_{t,T}}\int_0^\infty\left(2\xi+4\xi^2K^2\right)e^{\xi K^2}\Call^{vix}(t,K,T)dK\ .\]
Moreover, if the MGF $\E _te^{\xi \smallVIX_T^2} <\infty $ for some $\xi>0$, then
\begin{equation}
\label{eq:CvixLimit}\lim_{K\rightarrow \infty} K^2e^{\xi K^2}\Call^{vix}(t,K+\epsilon,T) =0\qquad\hbox{for all $\epsilon>0$ and for all $t\leq T$} \ .
\end{equation}
\end{proposition}

\begin{proof} It follows from Condition \ref{cond:int2} that $\VIX_T<\infty$ almost surely, and then integration by parts is used to check that 
\[e^{\xi \smallVIX_T^2} =1+   \int_0^\infty\left(2\xi+4\xi^2K^2\right)e^{\xi K^2}\left(\VIX_T-K\right)^+dK\ .\]
Then, for $t<T$ the monotone convergence theorem is used to obtain 
\begin{align*}
& \int_0^\infty\left(2\xi+4\xi^2K^2\right)e^{\xi K^2}\Call^{vix}(t,K,T)dK\\
 &=B_{t,T}\lim_{N\rightarrow\infty} \E _t \int_0^N\left(2\xi+4\xi^2K^2\right)e^{\xi K^2}\left(\VIX_T-K\right)^+dK\\
 &= B_{t,T}\E _t \lim_{N\rightarrow\infty}\int_0^N\left(2\xi+4\xi^2K^2\right)e^{\xi K^2}\left(\VIX_T-K\right)^+dK\\
 &=B_{t,T}\left(\E _te^{\xi \smallVIX_T^2}-1\right) \ .
\end{align*}
Moreover, $\E _te^{\xi \smallVIX_T^2} <\infty$ if and only if $\int_0^\infty\left(2\xi+4\xi^2K^2\right)e^{\xi K^2}\Call^{vix}(t,K,T)dK<\infty$, in which case for $a$ large and finite $\epsilon>0$,
\[\int_0^{a+\epsilon}K^2e^{\xi K^2}\Call^{vix}(t,K,T)dK \geq \int_0^aK^2e^{\xi K^2}\Call^{vix}(t,K,T)dK +\epsilon a^2e^{\xi a^2}\Call^{vix}(t,a+\epsilon,T)\ ,\]
where the inequality uses the monotonically-decreasing property of the call option in $K$. Taking the limit as $a$ tends toward infinity yields,
\[0\geq \overline{\lim}_{a\rightarrow\infty}a^2e^{\xi a^2}\Call^{vix}(t,a+\epsilon,T) \ ,\]
and since the quantity is non-negative it follows that the limit is zero.
\end{proof}

Section \ref{sec:modelMisprice} will show how the replication in Proposition \ref{prop:MGFhedge} is useful, as it can be part of an argument to show that reversal of the inequality in Lemma \ref{lemma:MGFbound} can result in static arbitrage. Moreover, this static replication is informative because it shows exponential decay for large-$K$ VIX call options if there is finiteness of $\E_te^{\xi\VIX_T^2}$ for some $\xi>0$. However, this large-$K$ asymptotic cannot be differentiated to get a tail distribution, but Section \ref{sec:EVT} will explore how to obtain the VIX's tail distribution.

%%%%%%%%%%%%%%%%%%%%%%%%%%%%%%%%%
\subsection{Robustness of Lemma \ref{lemma:MGFbound} to Mispricing}
\label{sec:modelMisprice}

Suppose there are mis-priced options such that there is a reversal of the inequality in \eqref{eq:MGFbound}. Lemma \ref{lemma:MGFbound} derives \eqref{eq:MGFbound} under the assumption that $\E$ is a risk-neutral expectation, and so there will be an arbitrage if such a mis-pricing occurs. In general, mis-pricings may not have an easily-identified arbitrage portfolio, but there are clearly defined, static portfolios that replicate all quantities considered in inequality \eqref{eq:MGFbound}, and  this section will show how mis-pricing can be exploited with a static portfolio of tradeable assets; the tradeable asset are futures and options on the underlying $S$ and $\VIX$, as well as shares in $S$.

Assume $S_tB_{0,t}$ is a martingale so that the future on $S_T$ is
\[F_{t,T} = S_t/B_{t,T}\qquad\forall 0\leq t<T\ .\]
For some $\xi,p$ and $q$ such that $\frac{2\xi p}{\tau}> 1$, let $\tilde \xi=\frac{2\xi}{\tau}$ and consider the following strategy at time $t<T$:
  \begin{align*}
& \mbox{-Buy portfolio $\Pi_t^1$ composed of European put options $\Put(t,\cdot,T+\tau)$,}\\
&~~~~~~~~~~\Pi_t^1=\frac{\tilde \xi(\tilde \xi q+1)}{B_{T,T+\tau}}\int_0^\infty \frac{\Put(t,K,T+\tau)}{K^{\tilde\xi q+2}}dK\ ,\\
 & \hbox{-Buy portfolio $\Pi_t^2$ composed of European call options $\Call(t,\cdot,T+\tau)$,}\\
  &~~~~~~~~~\Pi_t^2 =\frac{1}{B_{T,T+\tau}^{\tilde\xi p}}
  \tilde \xi(\tilde \xi p-1)\int_0^\infty \Call(t,K,T)\frac{dK}{K^{2-\tilde\xi p}}\ ,  \\
&  \hbox{-Sell portfolio $\Pi_t^3$ composed of VIX call options $\Call^{vix}(t,\cdot,T)$,}\\
&~~~~~~~~~\Pi_t^3
 =B_{t,T}+\int_0^\infty\left(2\xi+4\xi^2K^2\right)e^{\xi K^2}\Call^{vix}(t,K,T)dK\ .
 \end{align*}
It can be verified using integration by parts (see \cite{bick1982,carrLee2008,carrMadan2001,lee2004}) that these are replicating portfolios such that
\begin{align*}
\Pi_t^1&= \hbox{time-$t$ value of a claim on $\tfrac{1}{qB_{T,T+\tau}}S_{T+\tau}^{-\tilde\xi q}$, for all $t\leq T+\tau$, }\\ 
\Pi_t^2&=\hbox{time-$t$ value of a claim on  $\tfrac1p\left(F_{T,T+\tau}\right)^{\tilde \xi p}$, for all $t\leq T$}
 \end{align*} and as Proposition \ref{prop:MGFhedge} shows 
 \[\Pi_t^3=\hbox{time-$t$ value of a claim on $e^{\xi \VIX_T^2}$, for all $t\leq T$.}
 \]
In terms of these portfolios, inequality \eqref{eq:MGFbound} is equivalent to

\begin{equation}
\label{eq:pi_inequality}
\Pi_t^3<\Pi_t^1+\Pi_t^2\qquad\forall t\leq T\ .
\end{equation}
If inequality \eqref{eq:pi_inequality} is reversed then there will be a static arbitrage consisting of listed option prices, as concluded from the following two propositions:
 
 \begin{proposition}
 \label{prop:time_T_noArb}
 There is arbitrage at time $t=T$ if inequality \eqref{eq:pi_inequality} is reversed.
 \end{proposition}
 \begin{proof}
 \label{sec:atTime_T} 
 Let $\mathbb V_T$ denote the future (non-discounted) valuation of a claim. $\mathbb V_T$ is obtained from replication with listed prices for the underlying and options $\Put(T,K,T+\tau)$'s and/or $\Call(T,K,T+\tau)$'s (i.e. the integration-by-parts/replication method in \cite{bick1982,carrLee2008,carrMadan2001,lee2004}). For example, $\mathbb V_T[S_{T+\tau}] = F_{T,T+\tau}=\Call(T,0,T+\tau)$ and $-\frac{2}{\tau}\mathbb V_T[\log(S_{T+\tau}/F_{T,T+\tau})] = \VIX_T^2$. There is arbitrage at time $T$ unless the following strict inequality holds:  
 \begin{equation}
 \label{eq:violated_T_firstLeg}\Pi_T^3<\mathbb V_T\left[ e^{-\frac{2\xi}{\tau}\log(S_{T+\tau}/F_{T,T+\tau})}\right]\ .
 \end{equation}
Using the fact that $\Pi_T^3 = e^{\xi\smallVIX_T^2}$, if inequality \eqref{eq:violated_T_firstLeg} is reversed there is the following arbitrage portfolio:
 \begin{align*}
& \mbox{-Short $e^{\xi\smallVIX_T^2}$ claims on $-\frac{2\xi}{\tau}\log(S_{T+\tau}/F_{T,T+\tau})$,}\\
 & \hbox{-Long 1 claim on $(S_{T+\tau}/F_{T,T+\tau})^{-2\xi/\tau}$ (value given by right-hand side of \eqref{eq:violated_T_firstLeg}),}\\
&  \hbox{-Hold $(\xi\VIX_T^2-1)e^{\xi\smallVIX_T^2}$ contracts in the bond with price $B_{T,T+\tau}$,} 
 \end{align*}
 which has non-positive entry cost at time $T$, and at time $T+\tau$ has positive payoff due to the inequality $e^x-e^{x_0}-e^{x_0}(x-x_0)>0$ with $x = -\frac{2\xi}{\tau}\log(S_{T+\tau}/F_{T,T+\tau})$ and $x_0=\VIX_T^2$. 
 
Also at time $T$, there is an arbitrage unless the following strict inequality holds:
  \begin{equation}
 \label{eq:violated_T_secondLeg}\mathbb V_T\left[ e^{-\frac{2\xi}{\tau}\log(S_{T+\tau}/F_{T,T+\tau})}\right]<\Pi_T^1+\Pi_T^2 .
 \end{equation}
 Now using the fact that $\Pi_T^1$ is a claim on $\frac{1}{qB_{T,T+\tau}}S_{T+\tau}^{-\frac{2\xi q}{\tau}}$ and $\Pi_T^2$ is a settled claim equal to $\frac1p(F_{T,T+\tau})^{\frac{2\xi q}{\tau}}$, if equation \eqref{eq:violated_T_secondLeg} is reversed there is the following arbitrage portfolio:
 \begin{align*}
 & \mbox{-Short 1 claim on $  e^{-\frac{2\xi}{\tau}\log(S_{T+\tau}/F_{T,T+\tau}) } $,}\\
 & \hbox{-Long $B_{T,T+\tau}$-many of portfolio $\Pi_T^1$,}\\
&  \hbox{-Hold $\Pi_T^2$-many contracts in the $B_{T,T+\tau}$ bond,} 
 \end{align*}
  which has non-positive entry cost at time $T$. At time $T+\tau$ this portfolio has non-negative value and non-zero probability of positive payoff due to Young's inequality $e^{x+x_0}\leq \frac1qe^{xq}+\frac1pe^{x_0p}$ with $x=-\frac{2\xi}{\tau}\log(S_{T+\tau})$ and $x_0=\frac{2\xi}{\tau}\log(F_{T,T+\tau})$, and where there is equality iff $x=\tfrac{p}{q}x_0$. Hence, from equations \eqref{eq:violated_T_firstLeg} and \eqref{eq:violated_T_secondLeg} there is the strict inequality 
  \[\Pi_T^3<\Pi_T^1+\Pi_T^2\ .\]
 \end{proof}
 
  \begin{proposition}
  There is arbitrage at time $t<T$ if inequality \eqref{eq:pi_inequality} is reversed.
  \end{proposition}
  \begin{proof}
  If equation \eqref{eq:pi_inequality} is reversed, then there is a clear arbitrage by longing $\Pi_t^1$ and $\Pi_t^2$ and shorting $\Pi_t^3$ for a non-zero net cash flow at time $t$, which will then receive a positive cash flow at time $T$ because it was shown in Proposition \ref{prop:time_T_noArb} that $\Pi_T^2+\Pi_T^1-\Pi_T^3>0$.
  \end{proof}

%%%%%%%%%%%%%%%%%%%%%%%%%%%%%%%%%%%%%
\subsection{Extreme-Value Theory for VIX's Distribution}
\label{sec:EVT}

The moment-formula limits $\beta_R$ and $\beta_L$ have a relationship to the rates in the extreme-value distributions of the underlying asset. In particular, the stochastic volatility inspired (SVI) parameterization of the volatility surface leads to an extreme-value distribution that is parameterized by $\beta_R$ for the right tail and $\beta_L$ for the left. The SVI parameterization is related to extreme strikes because its construction is consistent with the moment formulas in \eqref{eq:moment1} and \eqref{eq:moment2}, namely that the square of large-strike implied volatility is proportional to log-moneyness divided by the square-root of time-to-maturity (see \cite{gatheralBook,gatheral2013}). The result of this section can be applied as a risk-neutral systemic risk indicator, which is related to those proposed in \cite{malz2013}.

This section's main technical hurdle is that the limiting behavior for VIX options given in Proposition \ref{prop:MGFhedge} cannot be differentiated in $K$. In other words, the limit in \eqref{eq:CvixLimit} shows a rate of convergence for $\Call^{vix}$, but cannot be differentiated to find an asymptotic for the tail distribution $\frac{\partial}{\partial K}C^{vix}$. However, by assuming differentiability in $K$ and assuming that the SVI is not mis-specified, the asymptotic tail distribution is obtained. In particular, the limiting right slope for VIX options $\beta_R^{vix}=\limsup_{K\nearrow\infty}\frac{\hat\nu^2(t,K,T)}{\log(K/\E_tS_T)/(T-t)}$ gives the rate of polynomial decay in the peaks-over-threshold (POT) distribution of $\VIX_T$ if $\beta_R^{vix}\in(0,2)$. 

The SVI uses a parameterization of the implied-volatility smile,
\begin{equation}
\label{eq:omega}
\omega(k) = a+b\left(\rho(k-m)+\sqrt{(k-m)^2+\sigma^2}\right)
\end{equation}
where $k = \log(K/X_{t,T})$, $X_{t,T}=\E_t\VIX_T$, and with parameters $(a,b,\rho,m,\sigma)$ fitted so that 
\[\omega(k) = \hat\nu^2(t,K,T)(T-t)\ ,\]
where $\hat\nu$ is VIX-option implied volatility given in Definition \ref{def:hatNu}. The VIX call price is then
\begin{equation}
\label{eq:callSVIprice}
\Call^{vix}(t,K,T) = B_{t,T}X_{t,T}\left(\Phi\left(d_+(k)\right)-\Phi\left(d_-(k)\right)e^{k}\right)
\end{equation}
where $d_\pm(k)= -\frac{k}{\sqrt{\omega(k)}}\pm\tfrac{\sqrt{\omega(k)}}{2}$.

The generalized Pareto distribution (GPD) with parameter $\alpha>0$ has cumulative distribution function defined  as
\[ G_\alpha(y)\triangleq \Bigg\{ 
\begin{array}{ll}
&\\[-1.1cm]
1-\left(1+y/\alpha\right)^{-\alpha}&\hbox{for $y\geq0$ and $\alpha<\infty$}\\[-0.32cm]
1-e^{-y}&\hbox{for $y\geq 0$ and $\alpha=\infty$}\ . 
\end{array}\]
The Pickands-Balkema-de Haan theorem states that the following are equivalent:
\begin{enumerate}[(i)]
\item The VIX's distribution function is in the maximum domain of attraction of $H_\alpha$ defined as $H_\alpha(y) = \exp(G_\alpha(y)-1)$ (see Appendix \ref{app:mda}).
\item There exists a positive measurable function $a(\cdot)$ such that the peaks-over-threshold (POT) distribution converges to the GPD in the following way
\[\lim_{x\nearrow\infty}\mathbb P_t\left(\frac{\VIX_T-x}{a(x)}\geq y\Big|\VIX_T\geq x\right)=1-G_\alpha(y)\ .\]
\end{enumerate}
For details on this theory and more general information about extreme-value theory, the reader is directed to \cite{degan, resnick1987}. Assuming twice-differentiability in $K$, the Breeden-Litzenberger formula yields the VIX's distribution function, $\mathbb P_t(\VIX_T\geq K) = -\tfrac{1}{B_{t,T}}\frac{\partial}{\partial K}\Call^{vix}(t,K,T)$. Then for two large strikes $K_1>K_0$ the POT distribution can be written as
\begin{align}
\label{eq:POTratio}
\mathbb P_t\left(\VIX_T\geq K_1\Big|\VIX_T\geq K_0\right)= \frac{\mathbb P_t(\VIX_T\geq K_1)}{\mathbb P_t(\VIX_T\geq K_0)} 
&=\frac{\int_{K_1}^\infty\frac{\partial^2}{\partial K^2}\Call^{vix}(t,K,T)dK}{\int_{K_0}^\infty\frac{\partial^2}{\partial K^2}\Call^{vix}(t,K,T)dK}\ .
\end{align}
If the VIX options imply a value $\beta_R^{vix}\in(0,2)$, then it follows that there is $\tilde p\in(0,\infty) $ such that $\E_t\VIX_T^{1+\tilde p}=\infty$, and the VIX's distribution is heavy tailed. However, the asymptotics from the moment formulas do not apply to the derivative $\frac{\partial^2}{\partial K^2}\Call^{vix}(t,K,T)$, but $\beta_R^{vix}$ combined with the SVI is used to obtain a GPD for the underlying's tail distribution. 
\begin{proposition}
\label{prop:EVT}
Suppose the large-strike limit from equation \eqref{eq:moment1} is applied to VIX call options to obtain 
\[\frac{\hat\nu^2(t,K,T)}{\log(K/X_{t,T})/(T-t)}\rightarrow \beta_R^{vix}\in(0,2)\qquad\hbox{as $K\rightarrow\infty$.}\]
Suppose further that the implied volatility surface is fit by the SVI parameterization, where the fit does not admit static arbitrage (i.e. no butterfly or calendar-spread arbitrage). Then the asymptotic behavior for the POT distribution with scaling function $a(x) = \tfrac x\alpha$ is

\[\lim_{x\nearrow\infty}\mathbb P_t\left(\frac{\VIX_T-x}{x/\alpha}\geq y\Big|\VIX_T\geq x\right)=\left(1+\frac{y}{\alpha }\right)^{-\alpha}\ ,\]
for $y\geq 0$ and $x$ tending toward infinity, where $\alpha =\frac{1}{2}\left(\sqrt{\frac{1}{\beta_R^{vix}}}+\frac{\sqrt{\beta_R^{vix}}}{2}\right)^2>1$ for all $\beta_R^{vix}<2$.
\end{proposition}
\begin{proof}
No butterfly arbitrage in the SVI fit is enough for there to be an explicit formula for the VIX distribution's density function at time $T$, namely
\[\frac{\partial }{\partial k}\mathbb P_t\Big(\log(\VIX_T/X_{t,T})\leq k\Big)= \frac{g(k)}{\sqrt{2\pi \omega(k)}}\exp\left(-\frac12d_-(k)^2\right)\ ,\]
where $k$, $d_\pm(k)$, $\omega(k) $ are given in equations \eqref{eq:omega} and \eqref{eq:callSVIprice}, and 
\[g(k) =\left(1-\frac{k\omega'(k)}{ 2\omega(k)}\right)^2-\frac{(\omega'(k))^2}{4}\left(\frac{1}{\omega(k)}+\frac14\right)+\frac{\omega''(k)}{2} \ .\]
The SVI fit to this slice of the implied-volatility surface is free from butterfly arbitrage if $g(k)\geq 0$ for all $k\in\mathbb R$ and $d_+(k)\rightarrow -\infty$ as $k\rightarrow\infty$ (see \cite{gatheralJaquier2014}).

The moment formula in equation \eqref{eq:moment1} says that $\omega(k) \sim k\beta_R^{vix}$ for $k$ large (i.e. as $K\nearrow\infty$), which applied to the density function yields
\begin{align*}
g(k)&\sim\frac14- \frac{(\beta_R^{vix})^2}{16}&\hbox{for $k$ tending toward $\infty$,}\\
d_-(k)&\sim-\sqrt{\frac{k}{\beta_R^{vix}}}-\frac{\sqrt{k\beta_R^{vix}}}{2}&\hbox{for $k$ tending toward $\infty$.}
\end{align*}
Placing these asymptotic approximations into the density yields
\begin{align*}
&\frac{\partial }{\partial k}\mathbb P_t\Big(\log(\VIX_T/X_{t,T})\leq k\Big)\sim \frac{(4-(\beta_R^{vix})^2)}{16\sqrt{2\pi k\beta_R^{vix}}}e^{-\alpha k}\qquad\hbox{for $k$ tending toward $\infty$,}
\end{align*}
where $\alpha =\frac{1}{2}\left(\sqrt{\frac{1}{\beta_R^{vix}}}+\frac{\sqrt{\beta_R^{vix}}}{2}\right)^2>1$ for all $\beta_R^{vix}<2$. Placing this asymptotic approximation of the tail density into equation \eqref{eq:POTratio}, and then applying L'Hopit\^al's rule, the result is found:
\[\frac{\mathbb P_t\left(\VIX_T\geq x\left(1+\frac{y}{\alpha }\right)\right)}{\mathbb P_t(\VIX_T\geq x)} \sim \frac{\int_{x\left(1+\frac{y}{\alpha }\right)}^\infty \frac{\partial^2}{\partial K^2}\Call^{vix}(t,K,T)dK}{\int_{x}^\infty \frac{\partial^2}{\partial K^2}\Call^{vix}(t,K,T)dK}\sim\frac{\int_{\ell(x,y)}^\infty \frac{1}{\sqrt k}e^{-\alpha k}dk}{\int_{\ell(x,0)}^\infty \frac{1}{\sqrt k}e^{-\alpha k}dk}\sim\left(1+\frac{y}{\alpha }\right)^{-\alpha}\]
for $x$ large, where $\ell(x,y) = \log(x(1+y/\alpha)/X_{t,T})$.
\end{proof}

Proposition \ref{prop:EVT} shows that for $\beta_R^{vix}\in(0,2)$ with an SVI fit, the POT distribution with scaling $a(x) =\tfrac x\alpha$ converges to a GPD with parameter $\alpha$. Hence, the Pickands-Balkema-de Haan theorem says this POT distribution is in the maximum domain of attraction of a generalized extreme-value distribution (see Appendix \ref{app:mda}). 
\begin{remark}Large-strike asymptotics for SVI are presented in \cite{sturm2010}. In particular, there is a large-$k$ series expansion in powers of $k^{-1/2}$.
\end{remark}
\begin{remark}If $\beta_R^{vix}=2$ then SVI may admit arbitrage as $d_+(k)\rightarrow 0$ as $k\rightarrow 0$, a limit which can be seen by using the large-$k$ SVI expansion in \cite{sturm2010}. If $\beta_R^{vix}=0$ then the distribution is not heavy tailed and the extreme-value theory needs to be reworked to find an appropriate scaling function $a(\cdot)$ for the POT distribution. 
\end{remark}
\begin{remark}The results of this section can also be applied to the SPX, namely $\beta_R$ and the SVI fit from SPX options can be used to obtain the GPD for tail of the SPX index.
\end{remark}

%%%%%%%%%%%%%%%%%%%%%%%%%%%%%%%%%%%%%%%%%%%%%%%%%%%%%%%%%%%%%%%%%%%%%%%%%%%%%%%%%%%%%%%%%%%%%%%%%
%%%%%%%%%%%%%%%%%%%%%%%%%%%%%%%%%%%%%%%%%%%%%%%%%%%%%%%%%%%%%%%%%%%%%%%%%%%%%%%%%%%%%%%%%%%%%%%%%
%%%%%%%%%%%%%%%%%%%%%%%%%%%%%%%%%%%%%%%%%%%%%%%%%%%%%%%%%%%%%%%%%%%%%%%%%%%%%%%%%%%%%%%%%%%%%%%%%
%%%%%%%%%%%%%%%%%%%%%%%%%%%%%%%%%%%%%%%%%%%%%%%%%%%%%%%%%%%%%%%%%%%%%%%%%%%%%%%%%%%%%%%%%%%%%%%%%
\section{Examples}
\label{sec:examples}
This section explores some widely-used stochastic volatility models. In particular the MGF of $\VIX_T^2$ and negative moments of $S_T$ are computed for several models; these calculations give insight on how the results from Section \ref{sec:extremeStrikes} apply. Models such as the 3/2 and SABR have heavy-tailed volatility processes and both have $\VIX_T^2$ MGFs that become infinite on the positive real line, but each have different moment formula asymptotics for VIX implied volatility. Other models in this section include the constant elasticity of volatility (CEV), the Heston, and the exponential Ornstein-Uhlenbeck (OU). From the September 9, 2010 data it is seen approximately that VIX call options have right-hand extreme strikes with $\sqrt{\beta_R} \approx .4495$, as shown in Figure \ref{fig:cev2momentAndFit}.

\subsection{CEV Model}
\label{sec:CEV}
The following example shows how special care is required in using models where the asset price does not satisfy Condition \ref{cond:moments}. Take $r=0$,  $T=1$, and the process

\[dS_t = \sqrt{S_t}dW_t=\sigma_tS_tdW_t\qquad t\in[0,1]\]
with $S_0=1$ and $\sigma_t =\sigma(S_t)= \frac{1}{\sqrt{S_t}}$. This example is informative because it is a non-trivial case where a barrier (i.e. via a stopping time) is used to ensure $\VIX_t<\infty$ a.s. Without the stopping time there would be positive probability of infinite VIX, which would lead to infinite variance-swap rates.

The process $S_t$ is a true martingale because $\int_\epsilon^\infty \frac{1}{s\sigma^2(s)}ds=\infty$ for some $\epsilon>0$ (see \cite{MU2012,protter}), but was shown in \cite{feller} to have positive probability of hitting zero in finite time, that is, $\mathbb P(\min_{t\leq 1}S_t=0)>0$ with zero being an absorbing boundary. Therefore $\VIX^2 =\E \int_0^1\frac{1}{S_u}du=\infty$ and $\E S_t^{-q}=\infty$ for all $t\in(0,1]$ and all $q>0$. However, in a manner similar to the real-life variance-swap barriers discussed in \cite{carrLee2009}, caps on payoffs based on realized variance can be modeled by using the stopping time 
\[\mathcal T = \min\left\{t\geq0\Big|\int_0^t\frac{1}{S_u}du=M\right\}\]
and then by considering the stopped process $\tilde S_t = S_{t\wedge\mathcal T}$. The VIX on $\tilde S$ is then bounded,
\[\VIX =\sqrt{\E \int_0^{\mathcal T\wedge1}\frac{1}{S_u}du}\leq \sqrt M\ ,\]
which is sufficient to have the Novikov condition for the martingales
\[\mathcal Z_t^\pm = \exp\left(-\frac12\int_0^{\mathcal T\wedge t}\frac{1}{S_u}du\pm\int_0^{\mathcal T\wedge t}\frac{1}{\sqrt{S_u}}dW_u\right)\ .\]
Hence, $\tilde S_t = \mathcal Z_t^+$ and for any $q\in(0,1]$ the negative moments are
\begin{align*}
\E \tilde S_t^{-q}&=\E\left[(\mathcal Z_t^+)^{-q}\right]\\
& =\E \exp\left(\frac{q}{2}\int_0^{\mathcal T\wedge t}\frac{1}{S_u}du-q\int_0^{\mathcal T\wedge t}\frac{1}{\sqrt{S_u}}dW_u\right)\\
& =\E \left[(\mathcal Z_t^-)^q\exp\left(q\int_0^{\mathcal T\wedge t}\frac{1}{S_u}du\right)\right]\\
&\leq e^{qM}\E \left[(\mathcal Z_t^-)^q\right]\leq e^{qM}\E \left[\mathcal Z_t^-\right]^q= e^{qM}\\
&<\infty\ , \end{align*}
and so Condition \ref{cond:moments} holds for $\tilde S_t$.

%%%%%%%%%%%%%%%%%%%%%%%%
\subsection{SABR Model}
\label{sec:SABR}
 This is another example of the contraposition of Proposition \ref{prop:negMoments_squareIntegrability} because Condition \ref{cond:int2} holds yet Condition \ref{cond:moments} does not because there are no negative moments on $S_t$. For $t\in[0,T+\tau]$, take $r=0$, $S_0=1$ and $Y_0>0$, and consider the SABR stochastic volatility model
\begin{align*}
d\log(S_t)& = -\frac12Y_t^2dt+Y_t\left(\sqrt{1-\rho^2}dW_t+\rho dB_t\right)\\
dY_t&=\alpha Y_tdB_t
\end{align*}
where $\alpha>0$, $W\perp B$, and $-1\leq \rho\leq -1$. It is well known that $S_t$ is a true martingale if and only if $\rho\leq0$ (see \cite{jourdain}), and for $\rho<0$ that $\E S_t^{1+p}<\infty$ if and only if $p\leq\frac{\rho^2}{1-\rho^2}$ (see \cite{jourdain,lionsMusiela2007}). The process $Y_t$ is log-normal and almost-surely positive, and clearly $\E Y_t^2<\infty$ implying that $\E \left(\int_0^tY_udB_u\right)^2<\infty$, which implies Condition \ref{cond:int2} and therefore the log-contract satisfies
\[-\E \log(S_t) = \frac{1}{2}\E \int_0^tY_u^2du<\infty\ .\]
Moreover, the VIX has all its moments
\begin{align*}
\E \VIX_T^{2n} &=\E \left(\E _T \frac{1}{\tau}\int_T^{T+\tau}Y_u^2du\right)^n\leq \frac{1}{\tau}\E \int_T^{T+\tau}Y_u^{2n}du\\
&= \frac{Y_0^{2n}}{\tau}\int_T^{T+\tau}e^{(2n^2-n)\alpha^2u}du\\
&<\infty\ ,
\end{align*}
yet has infinite MGF,
\[\E \exp\left(\xi\smallVIX_T\right)= \E \exp\left(\xi\sqrt{\frac{1}{\tau}\E _T\int_T^{T+\tau}Y_u^2du}\right)= \E \exp\left(\xi C_{\tau,\alpha}Y_T\right)=\infty\qquad \forall \xi>0\ ,\]
where $C_{\tau,\alpha} = \sqrt{\frac{1}{\tau Y_T^2}\E _T\int_T^{T+\tau}Y_u^2du}=\sqrt{\frac{1}{\tau\alpha^2}\left(e^{\alpha^2\tau}-1\right)}>0$. It follows that $\E e^{\xi\VIX_T^2}=\infty$ for all $\xi>0$. 

In summary, for $\rho<0$ the price $S_t$ is a martingale, there exists $\tilde p = \sup\{p\geq 0|\E S_{T+\tau}^{1+p}<\infty\}>0$, and $\E e^{\xi\VIX_T^2}=\infty$ for all $\xi>0$. Hence from equation \eqref{eq:MGFbound} it is deduced that $\tilde q=\sup\{q\geq 0|\E S_{T+\tau}^{-q}<\infty\} = 0$. The SABR model is good for fitting implied volatilities from SPX options data. However, SABR may have trouble fitting both SPX and VIX options because volatility modeled as geometric Brownian motion will produce a flat smile that does not have the upward skew observed from VIX options, i.e. $\VIX_T=C_{\tau,\alpha}Y_T$ and $C_{\tau,\alpha}\E_t(Y_T-K/C_{\tau,\alpha})^+$ leads to a flat implied-volatility smile because $Y_T$ is log-normally distributed.

This example has demonstrated the VIX-SPX relationship identified in Lemma \ref{lemma:MGFbound}, namely that $\E_te^{\xi\VIX_T^2}=\infty$ for all $\xi>0$ implies $\E_tS_{T+\tau}^{-q}=\infty$ for all $q>0$.

 %%%%%%%%%%%%%%%%%%%%%%%%
\subsection{CEV Volatility Process}
\label{sec:cev2vol}
This is an example of a model that is slightly better than SABR for pricing VIX options because it has implied-volatility function $\hat\nu(t,K,T)$ that is increasing in $K$. For $t\in[0,T+\tau]$, take $r=$0, $S_0=1$, and consider the stochastic volatility model
\begin{align*}
d\log(S_t)& = -\frac12Y_t^2dt+Y_t\left(\sqrt{1-\rho^2}dW_t+\rho dB_t\right)\\
dY_t&=cY_t^2dB_t
\end{align*}
where $c>0$ and $W\perp B$. The process $Y_t$ is almost-surely positive, a fact that can be deduced from its transition density (see the transition density in Appendix \ref{app:cev2}). 

Also from the transition density it is seen that $\E Y_t^p<\infty$ for $0\leq p\leq 3$ and $\E Y_t^p=\infty$ for $p>3$. Since  $\E Y_t^2<\infty$ it follows that $\E \left(\int_0^{T+\tau}Y_udB_u\right)^2<\infty$ and Condition \ref{cond:int2} holds, which means the log-contract satisfies
\[-\E \log(S_t) = \frac12\E \int_0^tY_u^2du<\infty\ .\]
If $\rho= 0$ or if $\rho<0$, then $S_t$ is a martingale (shown in Appendix \ref{app:cev2}); in terms of notation from Proposition \ref{prop:impVolLowerBound} this model also has finite, positive $\tilde\xi$ and $\tilde q$ (also shown in Appendix \ref{app:cev2}). 

\begin{figure}[htbp] %  figure placement: here, top, bottom, or page
\centering
\begin{tabular}{cc}
   \includegraphics[width=2.8in]{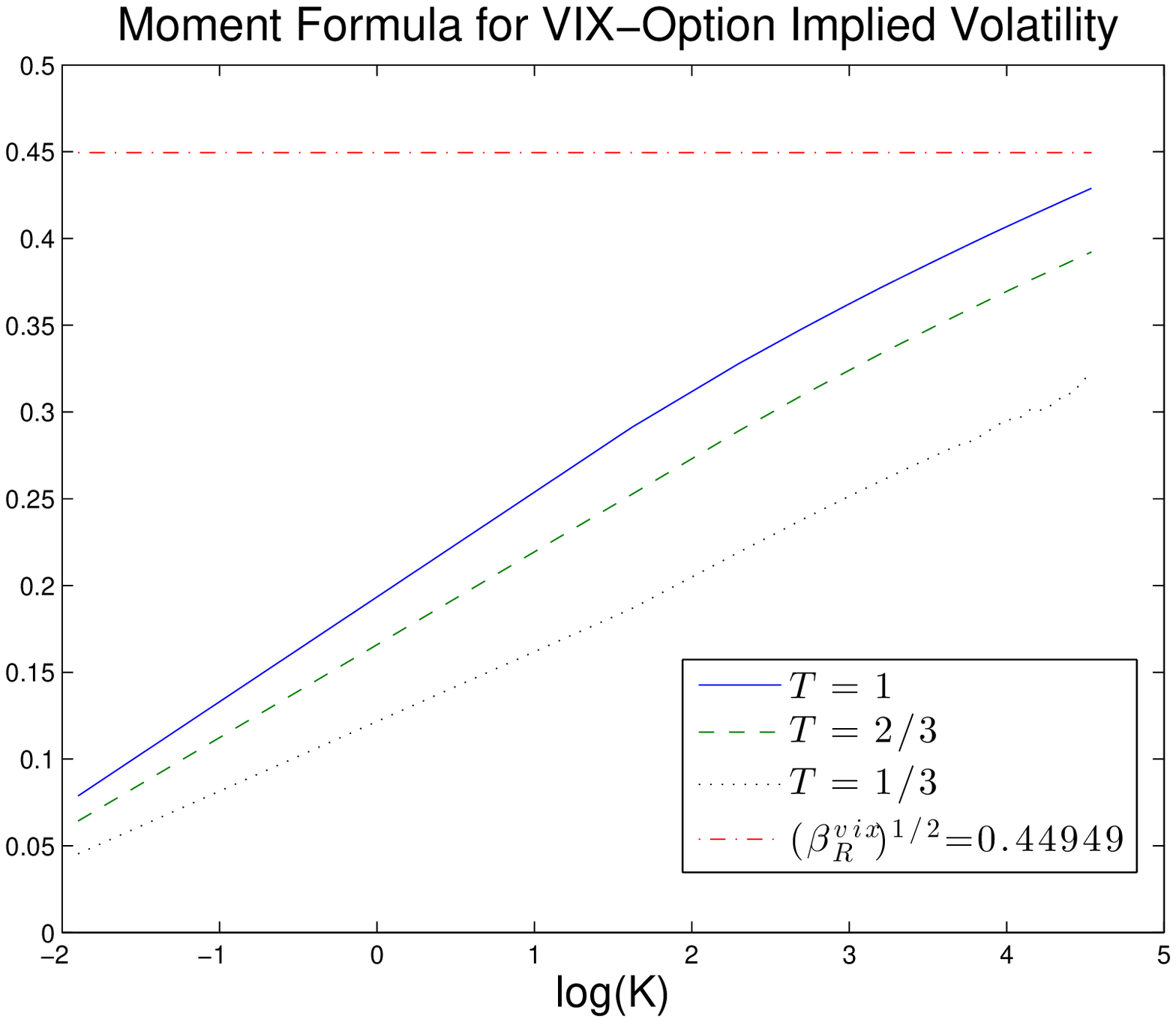} &
   \includegraphics[width=2.8in]{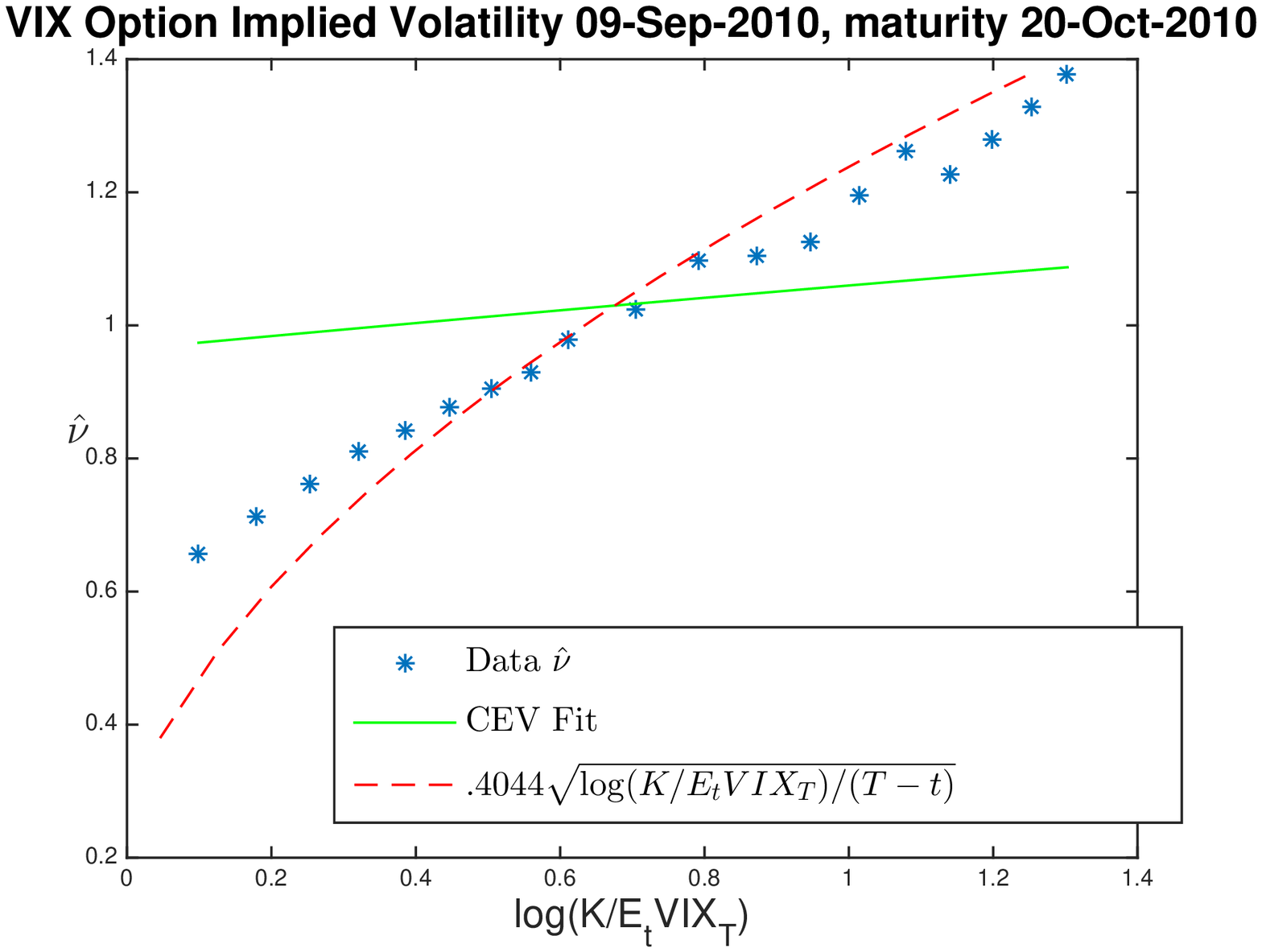} 
   \end{tabular}
   \caption{\textbf{Left:} Implied volatility for the CEV volatility process in Section \ref{sec:cev2vol}. The plot is of the scaled VIX implied volatility of September 9th, 2010, $\frac{\hat\nu(t,K,T)\sqrt{T-t}}{\sqrt{\log(K_{max}/\E _tS_T)}}\leq (\beta_R^{vix})^{1/2}$, for $Y_0 = .25$ and $K\leq K_{max} = 93.5$. The lines for $T=1,2/3$ and $1/3$ will come closer and closer to $\sqrt{\beta_R^{vix}} \approx 0.4495$ as $K_{max}$ increases. \textbf{Right:} The fit of the stochastic volatility model proposed in Section \ref{sec:cev2vol} is fit to the VIX options data from September 9th, 2010. The model captures some of upward skew in VIX implied volatility (certainly more than SABR which produces a flat implied volatility smile), yet cannot produce a steep enough slope to have a good fit to the data. The number $.4044$ is the lower bound on $\sqrt{\beta_R}$ that was identified in Figure \ref{fig:prelim_momentFormula}.}
   \label{fig:cev2momentAndFit}
\end{figure}

%%%%%%%%%%%%%%%%%%%%%%%%
\subsection{The Heston Model}
\label{sec:heston}
The Heston model is interesting because there is some interplay between the quantities $\tilde q$ and $\tilde\xi$ from Proposition \ref{prop:impVolLowerBound}. Furthermore, one of the model's most interesting features is the role played by time in determining negative moments, namely, for any $q>0$ there exists $t^*\in(-\infty,T)$ such that $\E_tS_T^{-q}=\infty$ for $t\leq t^*$. However, this section will show that the Heston model's time dependence in tail risk is (for the most part) absent in VIX options.

The Heston model (with $r=0$) has squared volatility given by a Cox-Ingersol-Ross (CIR) process, so that
\begin{align*}
d\log(S_t)&=-\frac12Y_tdt+\sqrt{Y_t}dW_t\\
dY_t&=\kappa(\overline Y- Y_t)dt+\gamma\sqrt{Y_t}dB_t
\end{align*}
where $dW_tdB_t=\rho dt$ for $\rho\in(-1,0)$. It is also important to have the Feller condition, $\gamma^2\leq 2\overline Y \kappa$. 

As shown in \cite{andersenPiterbarg2007,sturm2010}, if $(q\gamma \rho +\kappa)^2<\gamma^2q(1+q)$ then the earliest time $t^*<T$ such that $\E_{t^*}S_T^{-q}=\infty$ for some $q>0$ is 
\[T-t^* =\frac{1}{\sqrt{\gamma^2q(1+q)-(q\gamma \rho +\kappa)^2}}\left(\pi\indicator{q\gamma \rho +\kappa>0}+ \tan^{-1}\left(-\frac{\sqrt{\gamma^2q(1+q)-(q\gamma \rho +\kappa)^2}}{q\gamma \rho +\kappa}\right)\right)\ .\]

For the VIX, the MGF $\E_te^{\xi\smallVIX_T^2}<\infty $ for some $\xi>0$, but there is also a point $\tilde\xi$ of explosion that can be calculated explicitly. Using the SDE for $Y$, the squared VIX is linear in $Y_t$, $\VIX_t^2 = \frac1\tau\E_t\int_t^{t+\tau}Y_udu = \overline Y+\frac{Y_t-\overline Y}{\kappa \tau}\left(1-e^{-\tau\kappa}\right)=a+bY_t$. The explicit transition density for $Y_t$ (as given in \cite{aitSahalia1999}) is 
\begin{equation}
\label{eq:CIRtrans}
\frac{\partial}{\partial y}\mathbb P(Y_T\leq y|Y_t=y_0) = ce^{-u-v}(v/u)^{\alpha/2}I_\alpha(2\sqrt{uv})\ ,
\end{equation}
where $c = 2\kappa/(\gamma^2(1-e^{-\kappa(T-t)}))$, $u=cy_0e^{-\kappa(T-t)}$, $v=cy$, $\alpha= 2\overline Y\kappa/\gamma^2-1\geq0$, and $I_\alpha$ is the modified Bessel function of the first kind of order $\alpha$. The function $e^{\xi b y}$ is integrable against this density for any $\xi b<c$, hence, $\E_te^{\xi \smallVIX_T^2}<\infty$ if and only if $\xi<\tilde\xi$ where
\begin{equation}
\label{eq:tildeXiheston}
\tilde\xi =\frac cb= \frac{2\kappa^2\tau}{\gamma^2(1-e^{-\kappa(T-t)})\left(1-e^{-\tau\kappa}\right)}\ .
\end{equation}
Note that for $T-t$ large $\frac{\partial}{\partial y}\mathbb P(Y_T\leq y|Y_t=y_0) = ce^{-u-v}(v/u)^{\alpha/2}I_\alpha(2\sqrt{uv})\sim \frac{y^\alpha e^{-\frac{2\kappa}{\gamma^2 }y}}{\left(\frac{\gamma^2}{2\kappa}\right)^{\alpha+1}\Gamma(\alpha+1)}$ (in fact, $Y_t$ converges weakly to a gamma-distributed random variable). Hence for large time the VIX's MGF is
\[\E_t e^{\xi b Y_T}\sim \left(1-\frac{\gamma^2\xi b}{2\kappa}\right)^{-(\alpha +1)}\qquad\hbox{for }T-t\gg 1/\kappa\ ,\]
which exists for $\xi <\frac{2\kappa}{\gamma^2 b}=\frac{2\kappa^2\tau}{\gamma^2 \left(1-e^{-\tau\kappa}\right)}$, in agreement with \eqref{eq:tildeXiheston}.

The analysis above is interesting because it shows how the Heston model has a time dynamic in SPX tail risk, but which is not as present in the VIX's distribution. In other words, there is an increase in tail risk because $\E_tS_T^{-q}=\infty$ for some finite $T>t$, yet the VIX's MGF exists for a segment of the positive real line for $T>t$. This is striking because it means the Heston-model gives VIX prices that do not capture the same long-term risk that is in the underlying.

Figure \ref{fig:hestonFit} shows how the Heston model can fit the SPX implied volatility, but has some difficulty in fitting the VIX implied volatility. In particular, the CIR process of the Heston model leads to a downward slope in VIX option implied volatility, which is the stylistic feature pointed out in \cite{drimus2012,gatheralSlides, papanicolaouSircar2014} and shows the CIR process is not a good fit to the VIX data.

\begin{figure}[!htbp] %  figure placement: here, top, bottom, or page
\centering
\begin{tabular}{cc}
   \includegraphics[width=2.8in]{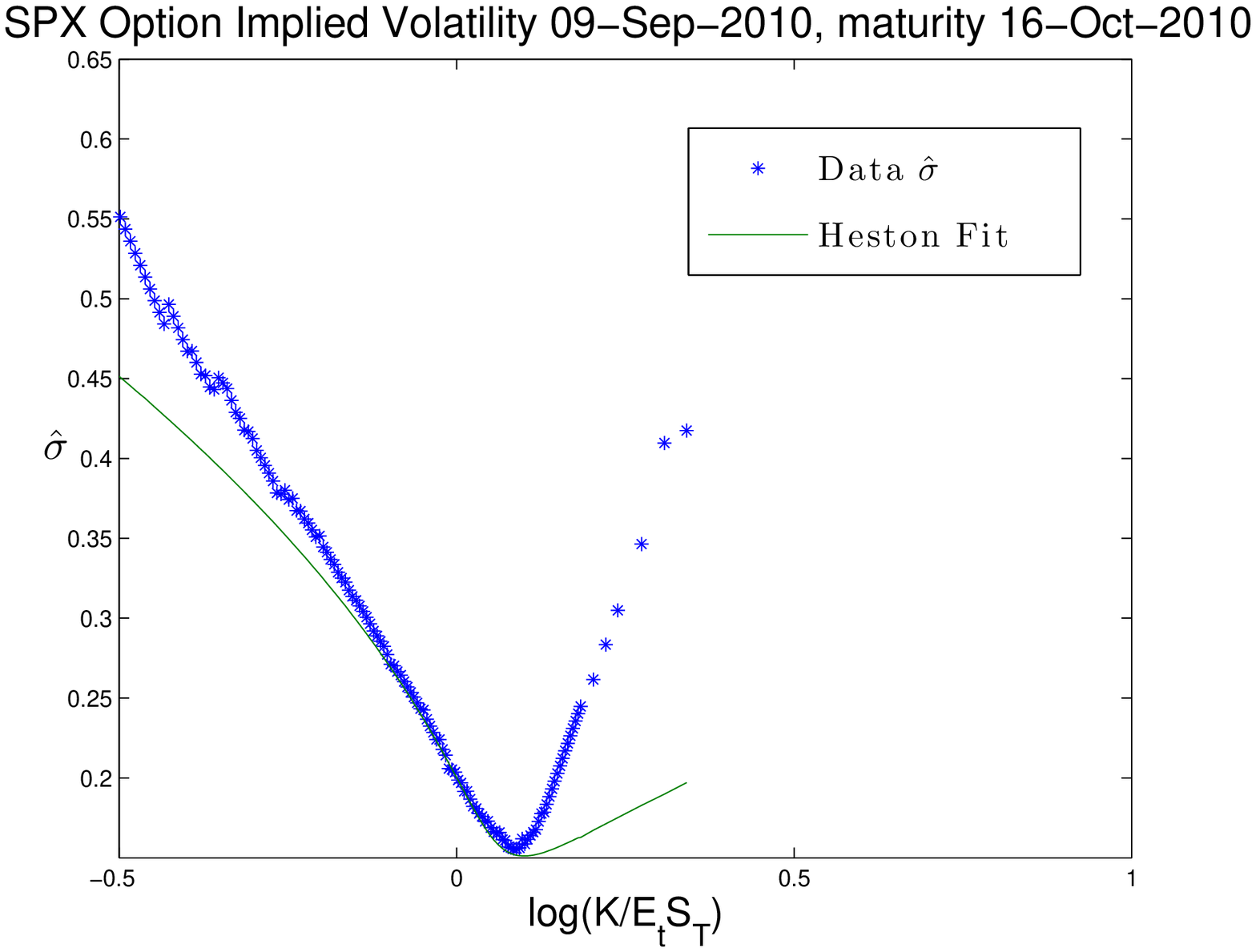} &
   \includegraphics[width=2.8in]{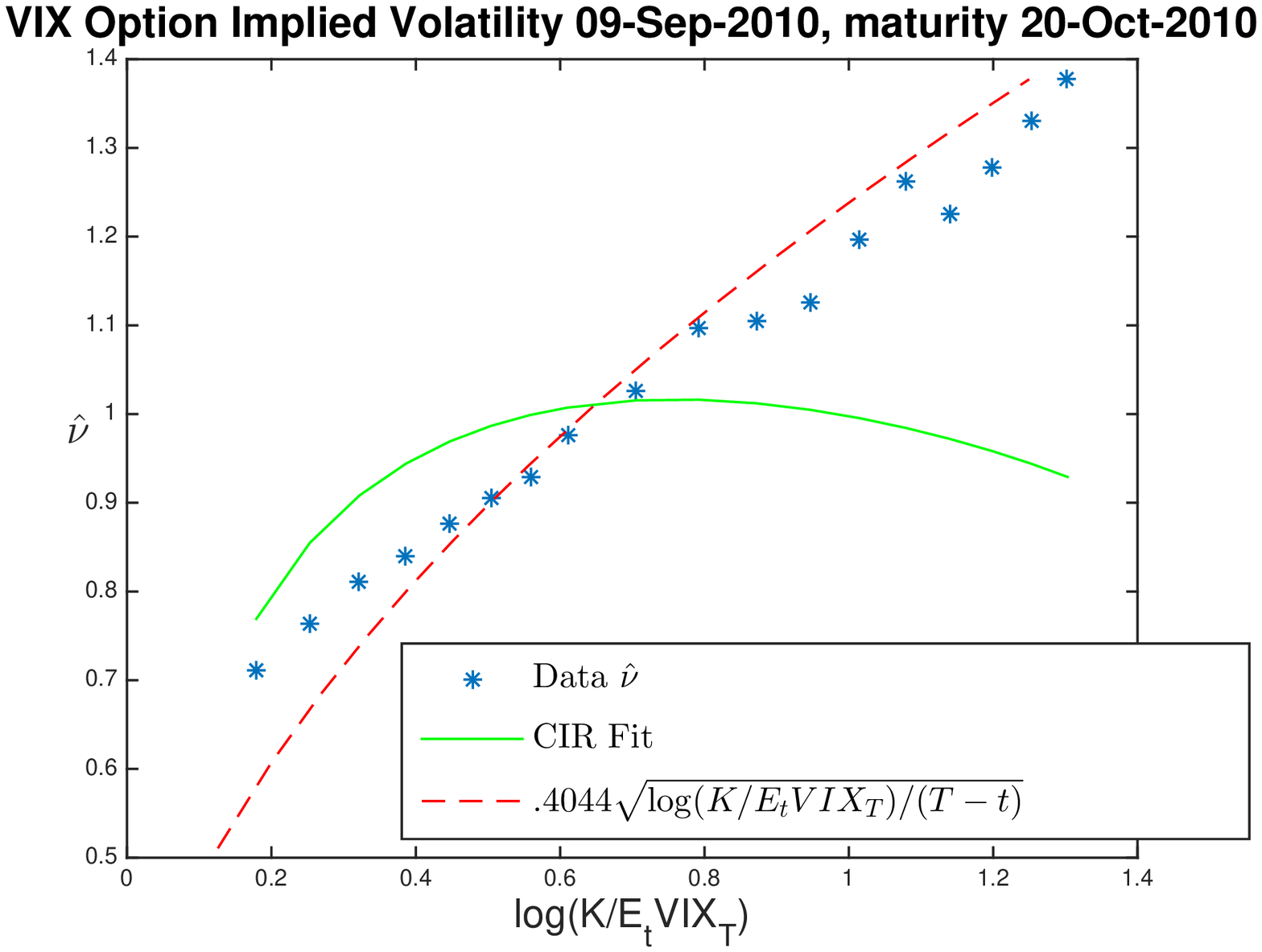} 
   \end{tabular}
   \caption{\textbf{Left:} A fit of the Heston model to the implied volatility smile from SPX options on September 9th, 2010. \textbf{Right:} A fit of the Heston model (or simply a CIR process) to the implied volatility from VIX options. The downward slope in the right-hand skew illustrates the CIR process' inability to fit VIX options on September 9th, 2010.}
   \label{fig:hestonFit}
\end{figure}

The Heston model has $\tilde\xi>0$ and $\beta_L<2$, yet part 1 of Proposition \ref{prop:impVolLowerBound} tends to be uninformative because $\sqrt{\left(\frac{4\tilde \xi}{\tau}\right)^2+\frac{4\tilde \xi}{\tau}}-\frac{4\tilde \xi}{\tau}$ is a small number for Heston calibrations having $\overline Y$ within the realm of what is (historically) observed in the data. In general, Proposition \ref{prop:impVolLowerBound} is model free, which means there is an underlying tradeoff between sharpness and specification of the model. Sharper estimates based on Proposition \ref{prop:impVolLowerBound} can be obtained for Heston and other specific volatility models.
%%%%%%%%%%%%%%%%%%%%%%%%
%%%%%%%%%%%%%%%%%%%%%%%%

\subsection{The Exponential Ornstein-Uhlenbeck (Exp-OU) Model} 
\label{sec:ExpOU}
The Exp-OU model is 
\begin{align*}
d\log(S_t)&=-\frac12e^{2Y_t}dt+e^{Y_t}dW_t\\
dY_t&=\kappa\left(\overline Y-Y_t\right)dt+\gamma dB_t\ ,
\end{align*}
where $\overline Y\in\mathbb R$, $\kappa,\gamma>0$, and $dW_tdB_t=\rho dt$ for $\rho\in(-1,1)$. Clearly $Y_t$ is Gaussian with $\E Y_t = Y_0e^{-\kappa t}+\overline Y(1-e^{-\kappa t})$ and $\hbox{Var}(Y_t) = \frac{\gamma^2}{2\kappa}\left(1-e^{-2\kappa t}\right)$, so Condition \ref{cond:int2} is satisfied,
\[\E\int_0^Te^{2Y_t}dt=\int_0^Te^{2\E Y_t+2\hbox{\small Var}(Y_t)}dt<\infty\ .\]
Moreover, all the moments of $\VIX_T^2$ exist, as
\[\E\VIX_T^{2n}=\E\left(\E_T\frac1\tau\int_T^{T+\tau}e^{2Y_t}dt\right)^n \leq\frac{1}{\tau}\E \int_T^{T+\tau}e^{2nY_t}dt<\infty\ ,\]
for all $n\geq 1$. However, the MGF of $\VIX_T^2$ does not exist for $\xi>0$,
\begin{align*}
\E \exp\left(\frac{\xi}{\tau}\E_T\int_T^{T+\tau}e^{2 Y_t}dt\right)&\geq\E \exp\left(\frac{\xi}{\tau}\int_T^{T+\tau}e^{2\E_T Y_t}dt\right)\\
&= \E \exp\left(\frac{\xi}{\tau}\int_T^{T+\tau}e^{2Y_Te^{-\kappa (t-T)}+2\overline Y(1-e^{-\kappa (t-T)})}dt\right)\\
&\geq \E \exp\left(\frac{\xi}{\tau}\int_T^{T+\tau}e^{2Y_Te^{-\kappa (t-T)}+2\overline Y(1-e^{-\kappa (t-T)})}dt\right)\indicator{Y_T>0}\\
&\geq \E \exp\left(e^{2Y_Te^{-\kappa\tau}}\left(\frac{\xi }{\tau}\int_T^{T+\tau}e^{2\overline Y(1-e^{-\kappa (t-T)})}dt\right)\right)\indicator{Y_T>0}\\
&=\infty\ ,
\end{align*}
because the MGF of log-normal random variable $e^{2Y_Te^{-\kappa\tau}}$ does not exist on the positive real line.

%%%%%%%%%%%%%%%%%%%%%%%%

\subsection{The 3/2 Model} 
\label{sec:3/2}
The 3/2 model is 
\begin{align*}
d\log(S_t)&=-\frac12Z_tdt+\sqrt{Z_t}dW_t\\
dZ_t&=Z_t\left(\kappa-(\kappa\overline Y-\gamma^2)Z_t\right)dt-\gamma Z_t^{3/2}dB_t\ ,
\end{align*}
where $dW_tdB_t=\rho dt$ for $\rho\in(-1,1)$, $\kappa>0$, $2\overline Y\kappa> \gamma^2$, and $\kappa \overline Y -\rho\gamma\geq \frac12\gamma^2$ so that the price process is a martingale (see \cite{bernard2017,drimus2012}), and hence Condition \ref{cond:int2} holds up to time $T+\tau$. This model can be equivalently written as $d\log(S_t)=-\frac{1}{2Y_t}dt+\frac{1}{\sqrt{Y_t}}dW_t$ where $Y$ is the square-root (CIR) process from Section \ref{sec:heston} and $Z_t=1/Y_t$. This is a popular choice for pricing VIX options because the volatility process has heavy tails (see \cite{badranBaldeaux2013,drimus2012}). A fit of the 3/2 model to VIX option implied volatility is shown in Figure \ref{fig:3on2Fit}. The upper bound $\tilde q$ such that $\E S_T^{-q}=\infty$ if $q\geq\tilde q$ is calculated in Appendix \ref{app:3halfsNegativeMoments} to be $\tilde q= \frac{\sqrt{1+\gamma^2\alpha^2}-1}{2}$. The following proposition calculates $\tilde \xi$ that is the maximum value for which there is existence of the MGF for $\VIX_T^2$:

\begin{proposition}
\label{prop:3/2MGF}
Let $Y_t$ denote the CIR process from Section \ref{sec:heston} having parameters $\kappa,\overline Y, \gamma$ with $\alpha = \frac{2\overline Y\kappa}{\gamma^2}-1$ (note that $\alpha$ is positive because it was assumed above that $2\overline Y\kappa> \gamma^2$). For $\xi>0$, 
\[\E e^{\xi\VIX_T^2}<\infty\ ,\]
if and only if  
$\xi<\tilde \xi =\frac{\gamma^2\tau\alpha(\alpha+1)}{2}$ (recall the notation $\tilde \xi$ from Proposition \ref{prop:impVolLowerBound}).
\end{proposition}
\begin{proof}(See Appendix \ref{app:proofProp3/2MGF}).
\end{proof}

In general,  the SABR model of Section \ref{sec:SABR}, the CEV volatility model of Section \ref{sec:cev2vol}, the Exp-OU model of Section \ref{sec:ExpOU}, and the 3/2 models are candidates for an improved fit to the VIX data because the volatility process has heavier tail, whereas the Heston model's CIR process does not capture the right-hand skew in VIX implied volatility. However, the Exp-OU and the 3/2 model appear to be the best for fitting to VIX-option implied volatility, as the SABR  model has flat smile for VIX options, and the CEV volatility model has relatively little skew for VIX options (see right-hand plot in Figure \ref{fig:cev2momentAndFit}). Figure \ref{fig:3on2Fit} shows evidence that the 3/2 model can capture some of the right-hand skew in the VIX data. The moment relationships from the examples of Sections \ref{sec:CEV} to \ref{sec:ExpOU} are summarized in Table \ref{tab:summary}.

\begin{figure}[htbp] %  figure placement: here, top, bottom, or page
   \centering
   \includegraphics[width=4.3in]{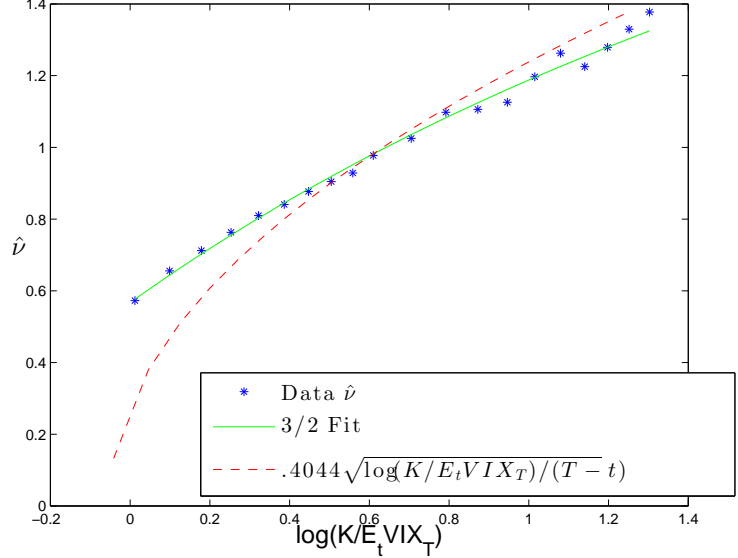} 
   \caption{The 3/2 model fit to implied volatility from VIX options. The heavy-tailed 3/2 process for volatility captures the increasing right-hand skew. The right-hand skew is the stylistic feature that is seen through the VIX data, and is important to capture when selecting a stochastic volatility model for pricing VIX options. See \cite{drimus2012, gatheralSlides,gatheral2013,papanicolaouSircar2014} for more insight on models that fit the right-hand skew. The fit uses approximated VIX formula: $\VIX_T^2\approx \frac{1}{\tau}\E_T\int_T^{T+\tau}\left[\frac{1}{ Y_T}-\frac{Y_{T+u}-Y_T}{Y_T^2}+\frac{(Y_{T+u}-Y_T)^2}{Y_T^3}\right]du$.}%; estimated CIR parameters: $(\kappa,\overline Y,\alpha,Y_0)=(20.84,  25.45,   8.92,   74.70)$.}
   \label{fig:3on2Fit}
\end{figure}

\begin{table}[htbp]
\centering
\begin{tabular}{|p{4cm}|p{2cm}|p{3.8cm}|p{3cm}| } 
 \hline
Model &$\E\VIX_T^p=\infty$&$\E e^{\xi \smallVIX_T^2} =\infty$  &$\E S_{T+\tau}^{-q}=\infty$\\
\hline
\hline
CEV model, $dS=S^adW$ with $0\leq a<1$&$\forall p>0$&$\forall\xi>0$&$\forall q>0$\\
 \hline
SABR with $\rho\leq 0$&$p=\infty$&$\forall\xi>0$&$\forall q>0$\\\hline
CEV volatility, $dS = S\sqrt YdW$ and $dY=cY^2dB$, with $dWdB=\rho dt$ and $\rho\leq 0$&$\forall p>3$&$\forall\xi> \frac{3\tau c^2}{2}$&$\forall q>\frac{\sqrt{1+c^2}-1}{2}$\\
\hline
 Heston Model&$p=\infty$&$\forall\xi\geq \frac{2\kappa^2\tau}{\gamma^2(1-e^{-\kappa T})\left(1-e^{-\kappa\tau}\right)}$&$\forall q>0$ such that $(q\gamma \rho +\kappa)^2<\gamma^2q(1+q)$ and $T$+$\tau\geq T^*(q)$\\
  [.2cm]\hline 
  $3/2$ model&$ p =\infty$&$\forall
  \xi\geq\frac{\gamma^2\tau\alpha(\alpha+1)}{2}$ 
  &$\forall q>\frac{\sqrt{1+\gamma^2\alpha^2}-1}{2}$;\\
  \hline
  Exp-OU model&$p=\infty$&$\forall \xi>0$&$\forall q>0$\\
  \hline
\end{tabular}
\caption{A table to summarize the moment relationships from the examples in Section \ref{sec:examples}. For the Heston model, the function $T^*(q)=\frac{1}{\sqrt{\gamma^2q(1+q)-(q\gamma \rho +\kappa)^2}}\left(\pi\indicator{q\gamma \rho +\kappa>0}+ \tan^{-1}\left(-\frac{\sqrt{\gamma^2q(1+q)-(q\gamma \rho +\kappa)^2}}{q\gamma \rho +\kappa}\right)\right)$; when $(q\gamma \rho +\kappa)^2<\gamma^2q(1+q)$ $T^*$ is the moment of explosion computed in \cite{andersenPiterbarg2007,sturm2010}. For the $3/2$ model, $\alpha= \frac{2\overline Y\kappa}{\gamma^2}-1>0$ as in Section \ref{sec:3/2}.} 
\label{tab:summary}
\end{table}

%%%%%%%%%%%%%%%%%%%%%%%%
%%%%%%%%%%%%%%%%%%%%%%%%

\subsection{Inclusion of Jumps} 
\label{sec:jumps}
This final example will show how the theory presented in earlier sections can be generalized to include an independent jump process. Specifically, an independent jump term is added to the log-returns model presented in equation \eqref{eq:returns}. Addition of these jumps will cause some slight changes to statements of Sections \ref{sec:framework} and \ref{sec:extremeStrikes}, but the main results of the paper remain intact. The main differences are that variance swaps will no longer be equal to the squared VIX (i.e. the relationship in equation \eqref{eq:noMultiplier} will change), and Proposition \ref{prop:negMoments_squareIntegrability} will need to be modified. 

Consider two more stochastic processes: a Poisson arrival process $N_t$ with intensity $\lambda\in(0,\infty)$, and an i.i.d. jump process $Y_i\in\mathbb R$ with $\E e^{Y_i}<\infty$ for $i=1,2,3,\dots$, where $N$ and $Y$ are independent of each other and jointly independent of $(W,\sigma)$. Equation \eqref{eq:returns} is modified to include jumps as follows:
\begin{equation*}
\label{eq:jumpReturns}
\log(S_t/S_0)=\int_0^t\left(r-\mu-\frac12\sigma_u^2\right)du+ \int_0^t\sigma_udW_u+\sum_{i=1}^{N_t}Y_i\ ,
\end{equation*}
where $\mu$ is a compensator so that $S_te^{-rt}$ is a local martingale. The quadratic variation of $\log(S_t)$ for the continuous-time model was $\int_0^t\sigma_u^2du$, but for this model is $\int_0^t\sigma_u^2du+\sum_{i=1}^{N_t}Y_i^2$. This model is part of the general class of L\'evy jump diffusions where variance swaps and the log contract differ by a jump premium. For the example presented above it was shown in \cite{carrWu} that the premium is an additive term, 
\[\hbox{variance-swap rate} =\frac1\tau\E_t\int_t^{t+\tau}\sigma_u^2du+\lambda \E Y_i^2 = \VIX_t^2 -2\lambda\E\left[e^{Y_i}-1-Y_i-\frac12Y_i^2\right]\ , \]
Alternatively, it is shown in \cite{carrLeeWu2011} that if log-returns can be written as a continuous time change of a L\'evy process $(X_t)_{0\leq t\leq 1}$, then the variance-swap rate is a multiplier of the log contract 
\[\hbox{variance-swap rate} =-\frac{Q_x}{\tau}\E_t\log(S_{t+\tau}/S_t)=\frac{Q_x}{2}\VIX_t^2\ , \]
where $Q_x = \frac{\hbox{\small Var}(X_1)}{\log\E e^{X_1}-\E X_1}$ is the multiplier. To generalize, Proposition \ref{prop:negMoments_squareIntegrability} is not too difficult to modify provided that $\log(S_t)$'s quadratic variation has been restated for jumps; independent jumps don't affect Lemma \ref{lemma:MGFbound}, the moment formulas, Proposition \ref{prop:impVolLowerBound}, Proposition \ref{prop:MGFhedge}, or Proposition \ref{prop:EVT}.
%%%%%%%%%%%%%%%%%%%%%%%%
\section{Conclusions}
\label{sec:conclusions}
This paper has explored some basic relationships between the markets for SPX and VIX options. The main idea is based on the notion that high-strike VIX call options can be used to hedge tail risk in the SPX. The moment formula was applied to relate the extreme-strike options on the SPX and VIX, and some formulas for comparison were introduced. The primary focus was the relationship between negative moments in SPX and the interval of the positive real line where the MGF of VIX-squared is finite. Negative moments and the MGF were computed for various stochastic volatility models, giving a sense of what can be accomplished with different models.

This paper is a step towards a theory for unifying option pricing of VIX and SPX contracts under a single stochastic volatility model. There are studies in the literature that have demonstrated better fits by richer models, but still there must be special care taken when using a single model for pricing both SPX and VIX options. The primary contribution of this paper is a model-free understanding of links between these two options markets. 

Possible future work will be refined arbitrage bounds for specific models, and more development of the model-free framework with the identification of no-arbitrage bounds that can be enforced with static hedging.

%%%%%%%%%%%%%%%%%%%%%%%%
\appendix

%%%%%%%%%%%%%%%%%%%%%%%%%%%%
%%%%%%%%%%%%%%%%%%%%%%%%%%%%
\section{3/2 Model}

\subsection{Proof of Proposition \ref{prop:3/2MGF}}
\label{app:proofProp3/2MGF}
It is shown in \cite{ahnGao1999} that
\[\E_01/Y_t = \frac{1}{\alpha}\zeta_te^{-Y_0\nu_t}\IFI(\alpha,1+\alpha,Y_0\nu_t)\ ,\]
where $\alpha= \frac{2\overline Y\kappa}{\gamma^2}-1>0$ as in Section \ref{sec:heston}, with
\begin{align*}
\zeta_t &= \frac{2\kappa}{\gamma^2(1-e^{-\kappa t})}\\
\nu_t&=\zeta_te^{-\kappa t}\ ,
\end{align*}
and where $\IFI(\alpha,1+\alpha,\nu)$ is the confluent hypergeometric function,
\[\IFI(\alpha,1+\alpha,\nu) = \frac{\Gamma(1+\alpha)}{\Gamma(\alpha)}\int_0^1e^{\nu r}r^{\alpha-1}dr=\alpha\int_0^1e^{\nu r}r^{\alpha-1}dr\ .\]
The moment generating function is,
\begin{align*}
\E_te^{\xi\smallVIX_T^2}& = \E_t\exp\left(\frac{\xi}{\tau}\int_T^{T+\tau}\E_T\frac{1}{Y_u}du\right)\\
& = \E_t\exp\left(\frac{\xi}{\alpha\tau}\int_T^{T+\tau}\zeta_{u-T}e^{-Y_T\nu_{u-T}}\IFI(\alpha,1+\alpha,Y_T\nu_{u-T})du\right)\\
& = \E_t\exp\left(\frac{\xi}{\tau} \int_0^1r^{\alpha-1}\int_T^{T+\tau}\zeta_{u-T}e^{-(1-r)Y_T\nu_{u-T}}dudr\right)\\
&=\E_t\exp\left(\frac{-2\xi }{\gamma^2\tau}\int_0^1r^{\alpha-1}Ei\left(-\frac{2(1-r)Y_T}{\gamma^2(e^{\kappa(u-T)}-1)}\right)\Big|_{u=T}^{T+\tau}dr\right)\ ,
\end{align*}
where $Ei(x) = -\int_{-x}^\infty \frac{1}{r}e^{-r}dr$ i.e. the exponential integral. Hence, letting $C$ be a constant that contains terms not involving $Y_T$, 
\begin{align}
\nonumber
&\E_te^{\xi\smallVIX_T^2}\\
\nonumber
&=\E_t\exp\left(\frac{2\xi }{\gamma^2\tau}\int_0^1r^{\alpha-1}\int_{\frac{2(1-r)Y_T}{\gamma^2(e^{\kappa\tau}-1)}}^\infty\frac{1}{u}e^{-u}dudr\right)\\
\nonumber
&\leq C \E_t\exp\left(\frac{2\xi }{\gamma^2\tau}\int_0^1r^{\alpha-1}\left(\int_{1\wedge\frac{2(1-r)Y_T}{\gamma^2(e^{\kappa\tau}-1)}}^1\frac{1}{u}du\right)dr\right)\\
\label{eq:3halfMGFVIXupperBound}
&= C\E_t\exp\left(\frac{2\xi }{\gamma^2\tau}\int_0^1r^{\alpha-1}\left(-\log\left(1\wedge\frac{2(1-r)Y_T}{\gamma^2(e^{\kappa\tau}-1)}\right)\right)dr\right)\\
\nonumber
&= C\E_t\exp\left(\frac{2\xi }{\gamma^2\tau}\int_{0\vee(1-\frac{\gamma^2(e^{\kappa\tau}-1)}{2Y_T})}^1r^{\alpha-1}\left(-\log\left(\frac{2(1-r)}{\gamma^2(e^{\kappa\tau}-1)}\right)-\log(Y_T)\right)dr\right)\\
\nonumber
&=C \E_t\left(Y_T\right)^{-\frac{2\xi }{\gamma^2\tau\alpha}\left(1-\left(0\vee(1-\frac{\gamma^2(e^{\kappa\tau}-1)}{2Y_T})\right)^\alpha\right)}\\
\nonumber
&\hspace{2cm}\times\exp\left(\frac{2\xi }{\gamma^2\tau}\int_{0\vee(1-\frac{\gamma^2(e^{\kappa\tau}-1)}{2Y_T})}^1r^{\alpha-1}\left(-\log\left(\frac{2(1-r)}{\gamma^2(e^{\kappa\tau}-1)}\right)\right)dr\right) \ .
\end{align}
Since it is assumed that $\alpha>0$ it follows that
\[\exp\left(\frac{2\xi }{\gamma^2\tau}\int_{0\vee(1-\frac{\gamma^2(e^{\kappa\tau}-1)}{2Y_T})}^1r^{\alpha-1}\left(-\log\left(\frac{2(1-r)}{\gamma^2(e^{\kappa\tau}-1)}\right)\right)dr\right)<\infty\qquad\hbox{a.s. }\ .\]
In addition, for $\xi>0$ it is straightforward to check that
\begin{align*}
&\E_t\left(Y_T\right)^{-\frac{2\xi }{\gamma^2\tau\alpha}\left(1-\left(0\vee(1-\frac{\gamma^2(e^{\kappa\tau}-1)}{2Y_T})\right)^\alpha\right)}
<\infty\\
&~~~~~~\hbox{iff}\\
&\E_t\left(Y_T\right)^{-\frac{2\xi }{\gamma^2\tau\alpha}}<\infty\ .
\end{align*}
Hence, the MGF of $\VIX_T^2$ is finite if $\E_t\left(Y_T\right)^{-\frac{2\xi }{\gamma^2\tau\alpha}}<\infty$. Moreover, it can be checked using the density of equation \eqref{eq:CIRtrans} that moments of $1/Y_t$ are infinite iff $\frac{2\xi }{\gamma^2\tau\alpha}\geq \alpha+1$ (i.e. iff $\xi\geq\frac{\gamma^2\tau\alpha(\alpha+1)}{2}$). Hence, the MGF of $\VIX_T^2$ is finite if $\frac{2\xi }{\gamma^2\tau}<\alpha(\alpha+1)$, which is a sufficient condition and shows one direction of the `iff' statement of the proposition.

In fact, $\frac{2\xi }{\gamma^2\tau}< \alpha(\alpha+1)$ is a necessary condition for finiteness of the MGF: there is a constant $C$ independent of $Y_T$ such that
\begin{align*}
&\E_te^{\xi\smallVIX_T^2}\\
&=\E_t\exp\left(\frac{2\xi }{\gamma^2\tau}\int_0^1r^{\alpha-1}\int_{\frac{2(1-r)Y_T}{\gamma^2(e^{\kappa\tau}-1)}}^\infty\frac{1}{u}e^{-u}dudr\right)\\
&= C \E_t\exp\left(\frac{2\xi }{\gamma^2\tau}\int_0^1r^{\alpha-1}\left(\int_{1\wedge\frac{2(1-r)Y_T}{\gamma^2(e^{\kappa\tau}-1)}}^1\frac{1}{u}e^{-u}du\right)dr\right)\\
&\geq C \E_t\exp\left(\frac{2\xi }{\gamma^2\tau}\int_0^1r^{\alpha-1}\left(\int_{1\wedge\frac{2(1-r)Y_T}{\gamma^2(e^{\kappa\tau}-1)}}^1\frac{1-u}{u}du\right)dr\right)\qquad\qquad (e^{-u}\geq 1-u\ ,~~\forall u\geq 0)\\
&= C\E_t\exp\left(\frac{2\xi }{\gamma^2\tau}\int_0^1r^{\alpha-1}\left(-\log\left(1\wedge\frac{2(1-r)Y_T}{\gamma^2(e^{\kappa\tau}-1)}\right)-\left(1-1\wedge\frac{2(1-r)Y_T}{\gamma^2(e^{\kappa\tau}-1)}\right)\right)dr\right)\\
&\geq C\E_t\exp\left(\frac{2\xi }{\gamma^2\tau}\int_0^1r^{\alpha-1}\left(-\log\left(1\wedge\frac{2(1-r)Y_T}{\gamma^2(e^{\kappa\tau}-1)}\right)-1\right)dr\right)\ ,
\end{align*}
which is, up to the term $-1$ in the integral, the same quantity as that in \eqref{eq:3halfMGFVIXupperBound}, and is infinite if $\frac{2\xi }{\gamma^2\tau}\geq \alpha(\alpha+1)$.

%%%%%%%%%%%%%%%%%%%%%%%%%
\subsection{Negative Moments}
\label{app:3halfsNegativeMoments}
It can also be shown that there are some negative moments of $S_T$. Following \cite{carrSun}, if the MGF of $\int_0^T\frac{1}{Y_t}dt$ exists then it is given by the formula
\begin{equation}
\label{eq:3on2RVmgf}
\E \exp\left(\xi \int_0^T\frac{1}{Y_t}dt \right)= \frac{\Gamma(b-a)}{\Gamma(b)}\left(\frac{2}{\gamma^2\phi}\right)^a \IFI\left(a,b,-\frac{2}{\gamma^2\phi}\right)\ ,
\end{equation}
where
\begin{align*}
\phi&=\frac{Z_0(e^{\kappa T}-1)}{\kappa}\\
a& = -\frac{\alpha}{2}+\frac12\sqrt{\alpha^2-\frac{8\xi}{\gamma^2}}\\
b &=2\left(\frac12+a+\frac{\alpha}{2}\right)\\
\alpha&=\frac{2\overline Y\kappa}{\gamma^2}-1\ . 
\end{align*}
This formula is real and positive if and only if
\[\xi\leq \frac{\gamma^2\alpha^2}{8}\ ,\]
and if the formula is complex then the MGF does not exist. Indeed, direct differentiation (verified with Mathematica) of equation \eqref{eq:3on2RVmgf} leads to the following asymptotic for the derivative of the MGF:
\[\frac{\partial}{\partial\xi}\E \exp\left(\xi \int_0^T\frac{1}{Y_t}dt \right)\sim \frac{C}{\sqrt{\frac{\gamma^2\alpha^2}{8}-\xi}}\qquad\hbox{as }\xi\nearrow \frac{\gamma^2\alpha^2}{8}\ ,\]
where $C$ is a constant. Hence, the MGF ceases to exist for $\xi > \frac{\gamma^2\alpha^2}{8}$ and the singularity occurs in the first derivative.

Next, for any $q>0$ apply It\^o's lemma to $S_t^{-q}$,
\[dS_t^{-q} = -kS_t^{-q}\frac{1}{\sqrt{Y_t}}dW_t+ \frac{q(q+1)}{2}S_t^{-q}\frac{1}{Y_t}dt\ ,\]
and define the stopping times 
\[\mathcal T_M = \inf\left\{t>0\Big|\int_0^t\frac{1}{Y_u}du\geq M\right\}\ ,\] 
which allows for the stopped process's stochastic integral to be a martingale on finite interval $[0,T]$, and hence,
\[\E S_{t\wedge\mathcal T_M}^{-q} = S_0^{-q} \E \exp\left(\frac{q(q+1)}{2}\int_0^{t\wedge\mathcal T_M}\frac{1}{Y_u}du\right)\qquad \forall t\leq T\ . \]
Hence,
\begin{align*}
&\lim_{M\rightarrow \infty} \E S_{T\wedge\mathcal T_M}^{-q}\\
&  = \lim_{M\rightarrow \infty} S_0^{-q}\E \exp\left(\frac{q(q+1)}{2}\int_0^{T\wedge\mathcal T_M}\frac{1}{Y_u}du\right)\\
&  =S_0^{-q}\E \lim_{M\rightarrow \infty} \exp\left(\frac{q(q+1)}{2}\int_0^{T\wedge\mathcal T_M}\frac{1}{Y_u}du\right)\qquad\hbox{monotone convergence}\\
&=S_0^{-q}\E\exp\left(\frac{q(q+1)}{2}\int_0^T\frac{1}{Y_u}du\right)\ .
\end{align*}
where the last term is the MGF of realized variance and was shown to be infinite for some $q>0$. Now notice the following convexity inequality,
\[\frac{1}{S_T^q} \geq \frac{1}{S_{T\wedge\mathcal T_M}^q} - q \frac{S_T-S_{T\wedge\mathcal T_M}}{S_{T\wedge\mathcal T_M}^{q+1}} \ ,\]
which yields 
\begin{align*}
\E S_T^{-q} &\geq \E\left[S_{T\wedge\mathcal T_M}^{-q} - q (S_T-S_{T\wedge\mathcal T_M})S_{T\wedge\mathcal T_M}^{-(q+1)}\right]\\
&= \E\left[S_{T\wedge\mathcal T_M}^{-q} - q \underbrace{\E[(S_T-S_{T\wedge\mathcal T_M})|Y_{t\wedge\mathcal T_M}]}_{=0}S_{T\wedge\mathcal T_M}^{-(q+1)}\right]\\
&=\E S_{T\wedge\mathcal T_M}^{-q}\\
&\rightarrow S_0^{-q}\E\exp\left(\frac{q(q+1)}{2}\int_0^T\frac{1}{Y_u}du\right)\qquad\hbox{as }M\rightarrow \infty\ .
\end{align*}
On the other hand, from Fatou's lemma,
\begin{align*}
\E S_T^{-q}& = \E\liminf_{M\rightarrow\infty} S_{T\wedge\mathcal T_M}^{-q}\\
& \leq \liminf_{M\rightarrow\infty} \E S_{T\wedge\mathcal T_M}^{-q}\\
&=S_0^{-q}\E\exp\left(\frac{q(q+1)}{2}\int_0^T\frac{1}{Y_u}du\right)\ .
\end{align*}
Hence, $\E S_T^{-q}=S_0^{-q}\E\exp\left(\frac{q(q+1)}{2}\int_0^T\frac{1}{Y_u}du\right)$ and $\E S_T^{-q}<\infty$ if and only if $\frac{q(q+1)}{2}\leq \frac{\gamma^2\alpha^2}{8}$, or in terms of the notation from Proposition \ref{prop:impVolLowerBound}
\[\tilde q= \frac{\sqrt{1+\gamma^2\alpha^2}-1}{2}\ .\]

%%%%%%%%%%%%%%%%%%%%%%%%%%%%%%%%%%%%%%%%%%%%%%%%%%%%%%%%%%%%%%%%%%%%%%%%%%%%%%%%%%%%%%%%%%%%%%%%%%%
%%%%%%%%%%%%%%%%%%%%%%%%%%%%%%%%%%%%%%%%%%%%%%%%%%%%%%%%%%%%%%%%%%%%%%%%%%%%%%%%%%%%%%%%%%%%%%%%%%%
\section{The CEV Model}
\label{app:cev2}
Consider the CEV model with quadratic variance,
\[dY_t = \sigma Y_t^2 dW_t\ .\]
This process is a strict local martingale and is discussed in \cite{coxHobson2005}. In particular, the process $X_t = 1/Y_t^2$ is among the class of SDEs considered in \cite{feller}, and has natural boundaries at zero and infinity (i.e. both $Y$ and $S$ have zero probability of touching zero). Furthermore, the transition density of $S$ is

\begin{align*}
\mathbb P(Y_T\in dz|Y_t = y) &= \frac{y}{z^3}\frac{dz}{\sqrt{2\pi(T-t)\sigma^2}}\\
&\hspace{.8cm}\times\left(\exp\left(-\frac{\left(\frac1z-\frac1y\right)^2}{2(T-t)\sigma^2}\right)-\exp\left(-\frac{\left(\frac1z+\frac1y\right)^2}{2(T-t)\sigma^2}\right)\right)\ .
\end{align*}
Table \ref{tab:cev2moments} shows the expectations of some important functions of $Y_T$ for $\sigma=1$. 

\begin{table}[h!]
\centering
\begin{tabular}{|c | p{10cm} |} 
 \hline
 $g(Y_T)$ & $\E \{g(Y_T)|Y_0 = y\}$\\ 
 \hline
 \hline
$Y_T$& $y \left(1 -2 \Phi\left(\frac{-1}{y\sqrt T}\right)\right) $,\\
$Y_T^2$&$\sqrt{\frac{2y^2}{T}}D_+\left(\frac{1}{y\sqrt{2T}}\right)$,\\ 
$\log(Y_T)$&$\frac 12\left(2+\gamma_e+\log(2/T)+\frac{\partial}{\partial a}\IFI\left(0,\frac 32,\frac{-1}{2Ty^2}\right)\right)$,\\[0.25cm] 
$(Y_T-K)^+$&$y\left(\Phi(\kappa-\delta)-\Phi(-\delta)+\Phi(\delta) -\Phi(\delta+\kappa) \right)$ $-K\left(\Phi(\kappa+\delta)-\Phi(\delta-\kappa)+\delta^{-1}\left(\Phi'(\kappa+\delta) -\Phi'(\kappa-\delta)\right) \right)$,\\
&$\delta=\frac{1}{y\sqrt T}$ $\kappa = \frac{1}{K\sqrt T}$\\
 \hline
\end{tabular}
\caption{The moments for the CEV process $dY =  Y^2dW$. The special functions are the normal Gaussian CDF $\Phi$, the Dawson integral $D_+=e^{-x^2}\int_0^xe^{u^2}du$, the confluent hypergeometric function of the first kind $\IFI=\frac{\Gamma(b)}{\Gamma(b-a)\Gamma(a)}\int_0^1e^{ux}u^{a-1}(1-u)^{b-a-1}du$, and the Euler Gamma $\gamma_e\approx 0.577215665$. For general $\sigma>0$ the transition density shows that $\sigma\neq 1$ is the scaling of time given by $\E ^\sigma[g(Y_T)|Y_0 = y]=\E ^1[g(Y_{T\sigma^2})|Y_0 = y]$.}
\label{tab:cev2moments}
\end{table}
%%%%%%%%%%%%%%%%%%%%%%%%%%%%%%%%%%%%%%%%%%%%%%%%%%%%%%%%%%%%%%%%%%%%%%%%%%%%%%%%%%%%%%%%%%%%%%%%%%%
%%%%%%%%%%%%%%%%%%%%%%%%%%%%%%%%%%%%%%%%%%%%%%%%%%%%%%%%%%%%%%%%%%%%%%%%%%%%%%%%%%%%%%%%%%%%%%%%%%%
%%%%%%%%%%%%%%%%%%%%%%%%%%%%%%%%%%%%%%%%%%%%%%%%%%%%%%%%%%%%%%%%%%%%%%%%%%%%%%%%%%%%%%%%%%%%%%%%%%%
%%%%%%%%%%%%%%%%%%%%%%%%%%%%%%%%%%%%%%%%%%%%%%%%%%%%%%%%%%%%%%%%%%%%%%%%%%%%%%%%%%%%%%%%%%%%%%%%%%%
%%%%%%%%%%%%%%%%%%%%%%%%%%%%

\subsection{CEV Volatility Process (Section \ref{sec:cev2vol})}
\label{app:CEVvol}

To determine whether or not $S_t$ is a true martingale it suffices to consider the case $S_0=1$ and the expectation
\begin{align*}
\E S_T&=\E \exp\left(-\frac 12\int_0^TY_u^2du+\sqrt{1-\rho^2}\int_0^TY_udW_u+\rho\int_0^TY_udB_u\right)\\
&=\E \exp\left(-\frac{\rho^2}{2}\int_0^TY_u^2du+\rho\int_0^TY_udB_u\right)\\
&=\E \left[\left(\frac{Y_T}{Y_0}\right)^{\rho/c}\exp\left(\frac{\rho(c-\rho)}{2}\int_0^TY_u^2du\right)\right]\ ,
\end{align*}
which is certainly a martingale if $\rho=0$. For $\rho<0$, define $\mathcal Z_t = \left(\frac{Y_t}{Y_0}\right)^{\rho/c}\exp\left(\frac{\rho(c-\rho)}{2}\int_0^tY_u^2du\right)$ and apply It\^o's lemma to get
\begin{align*}
\E S_T&=\E \mathcal Z_T= 1 +\E \left[ \rho\int_0^TY_u\mathcal Z_udB_u\right]=1\ ,
\end{align*} 
where the stochastic integral is a martingale because $ \E \int_0^T(Y_u\mathcal Z_u)^2du\leq  \int_0^T(\E Y_u^3)^{2/3}(\E\mathcal Z_u^6)^{1/3}du<\infty$ for $\rho<0$ and any $c>0$ (use the density in Appendix \ref{app:cev2} to verify that $\E Y_t^3<\infty$ and $\E Y_t^{-k}<\infty$ for all $k>0$), and hence $S_t$ is a true martingale. 

%%%%%%%%%%%%%%%%%%%%%%%%%
\subsubsection{MGF of $\VIX_T^2$}
Using the CEV process's 2nd moment form Table \ref{tab:cev2moments}, and taking constant $M >0$ and $\xi>0$, it is seen that the MGF of $\VIX_T^2$ can be broken into two terms, only one of which has the possibility of being infinite:
\begin{align*}
\E e^{\xi\smallVIX_T^2}&=\E \exp\left(\frac\xi\tau\E _T\int_T^{T+\tau}Y_u^2du\right)\\
&=\E \exp\left(\xi\frac{Y_T}{\tau}\int_T^{T+\tau}\sqrt{\frac{2}{c^2(u-T)}}D_+\left(\frac{1}{Y_T\sqrt{2c^2(u-T)}}\right)du\right)\\
&=\E \exp\left(\xi\frac{2}{\tau c^2}\int_{\tfrac{1}{Y_T\sqrt{2c^2\tau}}}^\infty \frac{1}{x^2}D_+(x)dx\right)\\
&=C_{M,\xi}\E \exp\left(\xi\frac{2}{\tau c^2}\int_{\tfrac{1}{Y_T\sqrt{2c^2\tau}}}^{\tfrac{1}{M\sqrt{2c^2\tau}}} \frac{1}{x^2}D_+(x)dx\right)\indicator{Y_T\geq M}\\
&\hspace{2cm}+\underbrace{\E \exp\left(\xi\frac{2}{\tau c^2}\int_{\tfrac{1}{Y_T\sqrt{2c^2\tau}}}^\infty \frac{1}{x^2}D_+(x)dx\right)\indicator{Y_T< M}}_{<\infty\qquad\forall M>0}\ ,
\end{align*}
where $C_{M,\xi} =  \exp\left(\xi\frac{2}{\tau c^2}\int_{\tfrac{1}{M\sqrt{2c^2\tau}}}^\infty \frac{1}{x^2}D_+(x)dx\right)<\infty$ for all $M>0$ and where there has been change of variable 
\begin{align*}
x &=\tfrac{1}{Y_T\sqrt{2c^2(u-T)}}\\
dx &= -\tfrac{1}{2Y_T\sqrt{2c^2}(u-T)^{3/2}}du = -\frac{Y_T\sqrt{2c^2}x^2}{2\sqrt{u-T}}du\ .
\end{align*}
Taking $M\geq 1/\sqrt{2c^2\tau}$ and examining the possibly infinite term, 
\begin{align}
\nonumber
& \E \exp\left(\xi\frac{2}{\tau c^2}\int_{\tfrac{1}{Y_T\sqrt{2c^2\tau}}}^{\tfrac{1}{M\sqrt{2c^2\tau}}}\frac{1}{x^2}D_+(x)dx\right)\indicator{Y_T\geq M}\\
\nonumber
&\geq\E \left[\exp\left(\xi\frac{2}{\tau c^2}\int_{\tfrac{1}{Y_T\sqrt{2c^2\tau}}}^{\tfrac{1}{M\sqrt{2c^2\tau}}}\left( \frac{1}{x}-\frac{2x}{3}\right)dx\right)\indicator{Y_T\geq M}\right]\\
\nonumber
&\geq e^{-\xi\frac{1}{3M^2(\tau c^2)^2}}\E \left[\exp\left(\xi\frac{2}{\tau c^2}\int_{\tfrac{1}{Y_T\sqrt{2c^2\tau}}}^{\tfrac{1}{M\sqrt{2c^2\tau}}} \frac{1}{x}dx\right)\indicator{Y_T\geq M}\right]\\
\nonumber
& = e^{-\xi\frac{1}{3M^2(\tau c^2)^2}}\E \left[e^{\xi\frac{2}{\tau c^2}\log(Y_T/M)}\indicator{Y_T\geq M}\right]\\
\label{eq:CEVvolLowerBound}
&=e^{-\xi\frac{1}{3M^2(\tau c^2)^2}} \E\left[\left(\frac{Y_T}{M}\right)^{\xi\frac{2}{\tau c^2}}\indicator{Y_T\geq M}\right]\ ,
\end{align}
where the inequality comes by using the Dawson-integral's MacLaurin series
\begin{align*}
D_+(x) &=\sum_{n=0}^\infty \frac{(-1)^n2^n}{(2n+1)!!}x^{2n+1}> x - \frac23x^3\qquad\hbox{for }0\leq x<1\ .
\end{align*}
The quantity in \eqref{eq:CEVvolLowerBound} is infinite if and only if $\xi\frac{2}{\tau c^2}>3$ because the CEV has infinite moments beyond 3 (this can be deduced by looking at the CEV density at the beginning of this appendix), and so a sufficient condition for the moment generating function of $\VIX_T^2$ to be infinite is
\[\xi> \frac{3\tau c^2}{2}\ .\]
This condition is also necessary, as the steps of \eqref{eq:CEVvolLowerBound} can be modified with the upper bound $\frac{1}{x^2}D_+(x)<1/x$ for $0<x<1$:
\begin{align}
\nonumber
& \E \exp\left(\xi\frac{2}{\tau c^2}\int_{\tfrac{1}{Y_T\sqrt{2c^2\tau}}}^{\tfrac{1}{M\sqrt{2c^2\tau}}}\frac{1}{x^2}D_+(x)dx\right)\indicator{Y_T\geq M}\\
\nonumber
&\leq\E \left[\exp\left(\xi\frac{2}{\tau c^2}\int_{\tfrac{1}{Y_T\sqrt{2c^2\tau}}}^{\tfrac{1}{M\sqrt{2c^2\tau}}} \frac{1}{x}dx\right)\indicator{Y_T\geq M}\right]\\
\nonumber
&= \E\left[\left(\frac{Y_T}{M}\right)^{\xi\frac{2}{\tau c^2}}\indicator{Y_T\geq M}\right]< \infty\qquad\hbox{if }~~~\xi\leq\frac{3\tau c^2}{2}\ .
\end{align}
Hence $\tilde \xi $ from Proposition \ref{prop:impVolLowerBound} is
\[\tilde \xi = \frac{3\tau c^2}{2}\ .\]

%%%%%%%%%%%%%%%%%%%%%%%%%
\subsubsection{Negative Moments}
For negative moments (for $q>0$) similar steps as those in the beginning of Appendix \ref{app:3halfsNegativeMoments} lead to
\begin{align*}
\E S_T^{-q} 
&= \E \exp\left(\frac{q+q^2}{2}\int_0^TY_t^2dt\right)\ .
\end{align*}
Next consider the process $X_t = Y_t^2$ and apply It\^o's lemma,
\[dX_t = c^2X_t^2dt+2cX_t^{3/2}dB_t\ ,\]
which is a 3/2 process, and if the moment exists then it is given the formula of \cite{carrSun},
\begin{align*}
\E S_T^{-q}=\E \exp\left(\frac{q+q^2}{2}\int_0^TX_tdt\right) 
&= \frac{\Gamma(\gamma-\alpha)}{\Gamma(\gamma)}\left(\frac{1}{2c^2TX_0}\right)^\alpha \IFI\left(\alpha,\gamma,\frac{-1}{2c^2TX_0}\right)\ ,
\end{align*}
where $\IFI$ is the confluent hypergeometric function and
\begin{align*}
\alpha&=-\frac14+\sqrt{\frac{1}{16}-\frac{q+q^2}{4c^2}}\\
\gamma&=2\left(\alpha+\frac34\right)\ .
\end{align*}
This is an analytic formula that is real and positive if $q+q^2\leq \frac{c^2}{4}$, and so $\E S_T^{-q}$ is finite if
\[q\leq\frac{-1+\sqrt{1+c^2}}{2}\ .\]
Sharpness of this finiteness inequality can be shown by following the same argument that was used to show the sharpness of negative moments condition in Section \ref{sec:3/2} and Appendix \ref{app:3/2MGF}. Hence, $\tilde q$ from Proposition \ref{prop:impVolLowerBound} is 
\[\tilde q = \frac{-1+\sqrt{1+c^2}}{2}\ .\]
%%%%%%%%%%%%%%%%%%%%%%%%%%%%
\section{Maximum Domain of Attraction}
\label{app:mda}
Let $F(y) = \mathbb P_t(\VIX_T\leq y)$. Consider samples $(Y_\ell)_{\ell=1}^n$ where $Y_\ell\sim iid~ F$ for each $\ell$, and let $M_n = \max(Y_1,Y_2,\dots,Y_n)$. For $\alpha>0$ the generalized extreme-value distribution is
\[H_\alpha(y) \triangleq \Bigg\{
\begin{array}{ll}
&\\[-1.1cm]
\exp\left(-(1+y/\alpha)^{-\alpha}\right)&\alpha<\infty\ ,\\[-0.35cm]
\exp(-\exp(-y))&\alpha=\infty \ .
\end{array}
\]
Distribution function $F$ is said to be in the maximum domain of attraction (MDA) of $H_\alpha$ for $\alpha<\infty$ if and only if 
\[a_nM_n \Rightarrow H_\alpha\qquad\hbox{as }n\rightarrow\infty \ ,\]
where $a_n= F^{-1}(1-1/n)$ (see \cite{degan} or \cite{resnick1987} page 54-57). This is written as $F\in \mbox{MDA}(H_\alpha)$. 

For $\alpha=\infty$, $F\in MDA(H_\alpha)$ if and only if
\[1-F(y)\sim \exp(-\Psi(y))\qquad\hbox{as $y\rightarrow\infty$}\]
for some function $\Psi\in C^2(\mathbb R^+)$ with (i) $\Psi(y)\rightarrow\infty$ as $y\rightarrow\infty$, (ii) $\Psi'(y)>0$, and (iii) $(1/\Psi'(y))'\rightarrow 0$ as $y\rightarrow\infty$ (see \cite{degan}).

%%%%%%%%%%%%%%%%%%%%%%%%%%%%%%%%%%%%%%%%%%%%%%%%%%%%%%%%%%%%%%%%%%%%%%%%%%%%%%%%%%%%%%%%%%%%%%%%%%%
%%%%%%%%%%%%%%%%%%%%%%%%%%%%%%%%%%%%%%%%%%%%%%%%%%%%%%%%%%%%%%%%%%%%%%%%%%%%%%%%%%%%%%%%%%%%%%%%%%%
%%%%%%%%%%%%%%%%%%%%%%%%%%%%%%%%%%%%%%%%%%%%%%%%%%%%%%%%%%%%%%%%%%%%%%%%%%%%%%%%%%%%%%%%%%%%%%%%%%%
%%%%%%%%%%%%%%%%%%%%%%%%%%%%%%%%%%%%%%%%%%%%%%%%%%%%%%%%%%%%%%%%%%%%%%%%%%%%%%%%%%%%%%%%%%%%%%%%%%%

\small{\bibliography{refs}}

\end{document}